\documentclass[12pt]{article}

\usepackage[colorlinks=true, citecolor=blue, linkcolor=blue, urlcolor=blue]{hyperref}
\usepackage{amsmath, amssymb, amsthm, mathtools, amsfonts}
\usepackage[pdftex]{color}
\usepackage[normalem]{ulem}
\usepackage{graphicx}
\usepackage{pgfplots}
\usepackage{subfig}
\usepackage{booktabs}
\usepackage{float}
\usepackage[linesnumbered]{algorithm2e}
\usepackage{csquotes}
\usepackage{dsfont}
\usepackage{paralist}

\usepackage[a4paper,twoside,margin=0.8in]{geometry}

\newtheorem{lemma}{Lemma}

\theoremstyle{definition}
\newtheorem{definition}{Definition}

\newtheorem{remark}{Remark}

\newtheorem{assumptions}{Assumptions}

\newtheorem{notation}{Notation}

\newtheorem{example}{Example}

\newcommand{\R}{\mathbb R}

\DeclareMathOperator{\argmax}{argmax}

\usepgfplotslibrary{fillbetween}
\renewcommand{\d}{\mathrm{d}}

\newcommand{\Ppharma}{P_{\text{p}}}
\newcommand{\Pinsur}{P_{\text{i}}}

\newcommand{\muA}{\mathcal{A}}
\newcommand{\muT}{\mathcal{T}}
\newcommand{\muO}{\mathcal{O}}

\newcommand{\tildepi}{\tilde\pi}

\theoremstyle{plain}
\newtheorem{theorem}{Theorem}

\DeclarePairedDelimiter\set\{\}

\pgfplotsset{compat=1.18}

\title{Who Pays, Who Benefits? \\ Producer–Insurer Games in Life-Saving Medicines}
\author{Delia Coculescu\thanks{Department of Mathematical Modeling and Machine Learning, University of Zurich. Email: delia.coculescu@df.uzh.ch}\thanks{Department of Finance, University of Zurich}, Maximilian Janisch\thanks{Institute of Mathematics, University of Zurich. Email: maximilian.janisch@math.uzh.ch}, Thomas Leh\'ericy\thanks{Department of Mathematical Modeling and Machine Learning, University of Zurich. Email: thomas.lehericy@math.uzh.ch}}

\begin{document}
	
	

	%
	
	%

	%

    \maketitle
	\begin{abstract}
		Pharmaceutical markets for life‑saving therapies combine monopoly power with insurance coverage. We build a tractable sequential game in which a patent‑holder chooses the drug price, a profit‑maximising insurer sets its premium, and a population of heterogeneous agents decide whether to insure and, conditional on diagnosis, whether to purchase treatment. Two sufficient statistics—subjective illness probability and reservation price—capture heterogeneity and nest risk-aversion and liquidity‑constraint motives within a unified framework. We prove existence of subgame‑perfect Nash equilibria and show that entry of an insurer strictly raises producer profits but may raise or lower both drug prices and treatment uptake, depending on the joint distribution of the population  statistics. Numerical experiments calibrated to flexible parametric families illustrate non‑monotone comparative statics and quantify conditions under which insurance reduces access. Our results provide benchmarks for evaluating price negotiations, price caps, and subsidy schemes in high‑cost drug markets.
	\end{abstract}
	
	
	
	\section{Introduction}
	Containing the cost of drugs and optimizing health outcome for populations require understanding the strategic interdependencies between the insurance industry, drug producers and patients. This is particularly true in market segments where competition between drug producers is absent, 
	and  fails to provide an incentive towards more affordable drug prices.  
	
	In this paper, we consider a life-threatening disease, such as cancer, and propose a competitive game involving a drug producer, a health insurer, and an infinite population of agents. We assume that the drug market is a monopoly---there is a unique, patent‐protected pharmaceutical product that has no close substitutes. For our model, we propose a leader-follower game, where the producer (acting as a leader) decides on the drug price first, and then the insurer (acting as a follower) responds by proposing an insurance premium to the population of agents. An agent in the population is modelled by certain characteristics which impact their decision-making under uncertainty. We are only interested in health-impacting decisions, such as buying insurance while healthy or following a treatment upon a diagnosis. Hence, once observing the drug price as well as the insurance premium, an agent can decide to buy insurance or not; in case of illness later on, an uninsured agent also gets to choose between buying the treatment out of pocket, or not. For simplicity, we only consider the case where the insurance contract fully reimburses the treatment, without co-payment.
	We study sub-game perfect Nash equilibria, and also equilibria of dictatorial type (where the insurer can impact the supply‐side pricing by threatening to exit the market), including the size of drug and insurance markets in equilibrium.  
	
	Our strategic game has perfect information in the following sense: the producer chooses his actions with full knowledge of the follower's decision criterion, and the insurer reacts after observing the producer's actions; at the end, the population reacts by observing the actions of the producer and insurer, and both the producer and the insurer perfectly anticipate the reaction of the population. However, the two big players (the drug producer and the insurer) do not necessarily have access to agent-level information in the population, but only to some sufficient statistical properties of the population. At the population level, we assume there is randomness and hence decision under uncertainty for each of the agents: randomness concerns an agent's future health status, as well as the drug efficacy in case of treatment, and in addition, uncertainty impacts the knowledge of probabilities of these hazards. In particular, we assume that there is an objective (statistical) probability, for instance the one derived from medical research on the incidence of the condition in a population and on the efficacy of the treatment, as well as subjective (agent-level) distortions of those probabilities by the agents, reflecting their beliefs, risk-aversion, and possibly resultant from additional agent-level information not explicitly modelled (such as risky or preventive lifestyles, genetics, etc). 
	
	Our model is aimed to model markets of drugs for life-threatening illnesses, such as cancer treatment drugs. This allows us for instance to omit moral hazard in the population (i.e., for an agent, the act of buying insurance does not subsequently magnify their probabilities of illness). The model is designed to be able to incorporate statistical properties of a population (such as incidence rate of an illness, socio-economic factors, as well as subjective factors that may be inferred from population surveys) and measures of drug efficacy, such as improvements in life expectancy or quality of life. 
	
	The landscape of drug regulation and the types of insurance plans for cancer drug treatments---or other life-threatening conditions---is extremely diverse from one country to the other, and often involves complex price negotiations. We do not attempt to model a specific country, disease, or treatment. Rather, our theoretical investigation can be seen as an exercise that sheds light on the incentives of the two big players, the producer and the insurer, in the absence of regulation. Our results can then serve as a benchmark in shaping effective regulations or negotiation processes (see Subsection \ref{sect:contribution}).

	\subsection{Characteristics of life-saving treatment markets}
	
	Life‑saving therapies (e.g., gene and cell treatments, orphan‑disease medicines, curative antivirals) are frequently beyond the ability of individuals to pay, representing one of the biggest burdens for modern healthcare and challenging equality and ethics.
	Below, we detail some particularities of these markets (quasi-monopoly and high willingness to pay of patients) as well as their implications on the resulting prices, focusing more specifically on cancer treatments. These aspects are incorporated to some degree into our model. 
	
	\textit{1. Quasi-monopoly of producers.} Cancer treatment markets  can be considered a monopoly (see \cite{SiddRajk12, KiraLe}). Indeed, most cancers are not curable, and treatments may work only for a limited time. When one treatment fails, the patient is treated with subsequent drugs until all options are exhausted. In fact, each drug is expected to be used in all patients during the course of their disease. In addition, new versions of older drugs become replacements and not alternatives, sustaining the monopoly.

	\textit{2. Exceptionally high willingness to pay.} A diagnosis that threatens survival creates in patients near‑inelastic demand. For many cancer drugs demand shows markedly lower price elasticities than other pharmaceutical classes  (\cite{Goldman23}).
	For curative gene therapies and ultra‑orphan drugs, payers have accepted launch prices well above one million per course of treatment\footnote{For instance, the estimated costs of the one-time treatments are 2.8 million for Zynteglo® and 2.2 million for Casgevy.}.

	\textit{3. Cancer treatments are costly, causing inequalities in access.} Launch prices routinely exceed USD 150'000--200'000 pee patient per year (\cite{Miljetal23}), an increasing trend over the last decades (\cite{Howard2015}).  Such prices mean that many cancer patients face severe financial hardship, labeled as ``financial toxicity'' (\cite{Gilligan}). In the United States, over one‑quarter of cancer patients delay, alter, or forgo therapy for cost reasons (see \cite{Ezekiel}); while low- and middle-income countries cover fewer that 15\% of eligible patients (\cite{Adams22, CardoneArnold23}).

	\textit{4. Weak correlation between price and therapeutic value.} Comparative studies across indications and jurisdictions reveal large price spreads that are poorly explained by R\&D cost, mechanism novelty, or incremental survival gains; some six‑figure therapies extend life by only weeks (\cite{Olivieretal21, Vokinger21, MailankodyPrasad, Miljetal23, YuHelmsBach}). On the other hand, the potential market size seems to have a large effect on pharmaceutical innovation \cite{AcemogluLinn04}.
	
	The elements listed above seem to indicate that without an efficient regulation, cancer drug manufacturers set prices based on ``what the market will bear'', exploiting their monopoly position.

	\subsection{Related literature}\label{sect:lit-review}
	Health-economic research has long emphasized how risk, liquidity constraints, and market power jointly shape access to life-saving therapies. Early normative work by \cite{Arrow1963} framed medical care as a market with pervasive uncertainty and justified broad risk‑pooling mechanisms.
	Subsequent positive models, starting with  \cite{Pauly1969},  stressed ex‑post moral‑hazard effects: once insured, patients may consume ``too much'' care. \cite{Nyman1999, Nyman2003}, introduced the ``access motive'': insurance expands welfare by relaxing liquidity constraints, hence is not only a risk reallocation device. While these contributions illuminate demand‑side forces, they all treat treatment prices as exogenous and therefore cannot speak to strategic pricing by manufacturers. 
	
	A second strand endogenises prices but focuses on a single powerful supplier. \cite{LakdawallaSood2009} and \cite{JenaPhilipson2013} show that if an insurer commits to cover a high‑value technology, a monopolist can appropriate much of the surplus, rendering seemingly cost‑effective treatments socially inefficient.  \cite{BESANKO2020102329} extend Nyman’s framework and demonstrate that even partial insurance (with co‑payments) raises monopoly prices; however, the insurer is still modelled as passive -- its premium is not chosen strategically and the population is homogeneous. 
	
	More recent work introduces bargaining or dynamic contracting, but usually abstracts from heterogeneity on the demand side. Examples include: (i) \cite{KyleQian2014} -- price ceilings with parallel trade in the EU; (ii) \cite{DuboisDeMouzonScottMortonSeabright2015} -- sequential entry of generics and price regulation; (iii) \cite{Grennan2013} and \cite{Gaynor2015} -- hospital‑device bargaining and insurer‑hospital bargaining, respectively, illustrating how buyer concentration disciplines prices.
	Although methodologically close (Stackelberg or Nash‑in‑Nash games), these papers study hospital–insurer or cross‑country negotiations rather than the producer--insurer--patient triad central to our approach.
	
	Finally, a growing empirical literature quantifies insurer reactions to pharmaceutical price -- e.g. \cite{DugganScottMorton2010} on Madicare Part D, \cite{EinavFinkelsteinPolyakova2018} on adverse selection in prescription-drug insurance -- but without modelling joint price setting.

	\subsection{Contribution of the present research}\label{sect:contribution}
	We study a life‑saving–drug market in which a monopolist producer and a profit‑maximising insurer interact before patients make coverage and consumption decisions.  The analysis is conducted under \emph{minimal regulation}: the insurer offers full coverage (no cost‑sharing), and the producer is free to set price.  
	Although in principle health insurance represents one of the crucial strategies for achieving social equality in the healthcare system and protect basic health rights, we show that—absent countervailing bargaining power or price controls— it can produce opposite effects.
	
	Our model displays some interesting equilibrium findings.  Comparing subgame‑perfect Nash equilibria with and without an insurer, we establish:
	
	\begin{compactitem}
		\item[-] \textit{Price effect.}  Drug prices \emph{usually} rise once the treatment is covered; nevertheless, calibrated counter‑examples reveal that prices can also fall. A lower price effect is usually not present in models without heterogeneous agents, or  strategic insurer. 
		\item[-] \textit{Access effect.}  Coverage need not increase; the share of diagnosed patients treated may even decline.   
		\item[-] \textit{Distributional effect.}  Producer profit is weakly higher with insurance. As insurer profit is non‑negative it follows that the average patient pays more in expectation—often without gaining better access.  This is contrary to the popular idea that insurance facilitates access to medication. 
		\item[-] \textit{Role of heterogeneity.}  It is the characteristics of the population---such as the level and repartition of the risk aversion, socio-economic variables, subjective probabilities that the individuals attach to a future diagnosis---driving the agents' willingness to pay that determine if the outcome with insurance is worse off or better of than the one without insurance, it terms of access to medication.
	\end{compactitem}
	
	Relative to the existing literature, the paper adds some methodological contributions. :
	\begin{enumerate}
		\item \textit{Endogenous two‑sided pricing.}  Producer price and insurer premium are chosen sequentially by strategic, profit‑seeking agents, not imposed exogenously or set by a passive insurer.
		\item \textit{Parsimonious micro‑foundation of heterogeneity.}  Patient diversity is captured by two sufficient statistics—the subjective illness probability $p$ and the reservation price $\psi$—derived from a dynamic life‑cycle model (Appendix~A).
		\item \textit{Unified positive and normative analysis.}  We prove that producer profits never fall when insurance is offered, whereas prices and access can move in either direction.  The framework therefore provides a benchmark for evaluating price‑negotiation mandates, value‑based subsidies, and other regulatory tools.
	\end{enumerate}
	
	By combining insurer strategy, monopoly power, and micro‑founded heterogeneity in a tractable game‑theoretic setting, the paper bridges the gap between traditional demand‑side insurance models and the growing literature on strategic pharmaceutical pricing.  Numerical results in Sections~\ref{sec:example}–\ref{sect:numerical-results} illustrate the policy trade‑offs under different parameterisations.
	
	These points highlight some key  outcomes; our model can exhibit a wider range of behaviors depending on the specific population characteristics and parameters. We refer the reader to Section~\ref{sec:example} and Section~\ref{sect:numerical-results} for a more detailed exploration of these effects through specific examples.

	\section{The model}\label{sect:math-model}
	
	We present a model for the interaction between two major market players ---a pharmaceutical company and an insurance provider--- serving a large population of rational individuals who represent potential treatment recipients. This population faces a constant incidence rate $r \in[0,1]$ for a specific disease of interest during a reference period, where for simplicity we assume affected individuals are immediately diagnosed, so that $r$ may also be interpreted as the statistical probability of a future diagnosis for currently unaffected individuals.  Also, we assume that agents become affected by the disease independently from other agents (thus excluding infectious diseases).
	Our timeline spans from time 0 (present) to time 1 (e.g., one year later). A multi-period generalization is presented in  Appendix \ref{sec:multiperiod}.  
	
	At the outset (time 0), all individuals are undiagnosed. The pharmaceutical company ---hereafter called the producer--- makes a treatment available for purchase at any point $t \in (0,1]$ at a price $\theta$, which is established at time 0. Concurrently, an insurance company ---hereafter called the insurer--- offers coverage for the treatment cost at a premium $\tilde{\pi} := \pi(\theta)$ for individuals diagnosed during the reference period.
	We only consider here the case where the insurance fully covers the treatment, without co-payment (i.e. the insurance pays the same price as a consumer), and assume the insurance prices are unique for the whole population\footnote{In practice, insurance prices may be dependent on age, or some other factors. Here, we ignore here such aspects. Similarly, we also assume a unique treatment price.}. Based on the observed treatment price $\theta$ and insurance premium $\tildepi$, each individual decides at time 0 whether to purchase insurance. Those who decline insurance coverage may later face a decision upon diagnosis: whether to purchase the treatment directly.
	
	The market interaction is structured as a sequential game at time 0, where the producer first establishes a price $\theta$ for the treatment, followed by the insurer's response in setting the premium. The insurer's strategies (denoted by $\pi$) are functions that map each potential drug price to a corresponding non-negative insurance premium value (denoted by $\tilde{\pi}=\pi(\theta)$). For an observed pair $(\theta, \tildepi)$, each agent individually is taking a health-impacting decision, during the time interval $[0,1]$, as detailed next. 
	
	\subsection{The population of agents}\label{sect:Customers}
	
	We assume that each agent in a large population  observes at time 0 the couple $(\theta,\tilde \pi)$, corresponding to the price of treatment and insurance premium. Each agent makes a choice (that is, a succession of actions during the reference period) among the three below: 
	\begin{center}\small
		\begin{tabular}{c|l|l}
			\textit{Choice} & \textit{At time 0} & \textit{Within time interval (0,1], upon diagnosis}\\
			\texttt{T}& no action & buy treatment at price $\theta$ \\
			\texttt{A}& buy insurance at price $\tilde\pi$ & receive treatment, reimbursed by insurance\\
			\texttt{O} & no action & no treatment
		\end{tabular}
	\end{center}
	Relevant for the subsequent analysis is the fraction of agents in each of the following three non-overlapping groups:  the agents that strictly prefer insurance  coverage at time 0, agents that decline coverage and prefer buying the treatment out-of-pocket, and, finally, the remainder of the agents, that will have no access to treatments, should they become ill. The identities of the agents are irrelevant for the insurer-producer game. 
	
	Therefore, we propose a reduced-form, yet sufficiently general formalization of the decision-making patterns in the population.  In  practice, it  can be obtained from more complex settings, some of which will be emphasized below. 
	In the reduced-form setting, an agent is represented by an element $ (p,\psi)\in [0,1]\times \mathbb{R}_+$:
	\begin{compactitem}
		\item[-] The first dimension $p$ represents an individual's subjective probability assessment at time 0 of being diagnosed during the reference period $(0,1]$, conditionally on being undiagnosed at time 0. This personal probability estimate may differ from the actual incidence rate $r$.
		\item[-] The second dimension $\psi$ represents an individual's reservation price for the treatment if diagnosed. This value indicates the maximum amount an individual is willing or able to pay out-of-pocket for the treatment upon diagnosis, should a diagnosis occur during the reference period, possibly  discounted to reflect its present value at time 0. 
	\end{compactitem}
	In such a reduced-form parameterization,  entire population is represented by a probability distribution $\mu$ defined on the parameter space $[0,1] \times \mathbb{R}_+$ (equipped with the Borelian sigma-algebra).
	
	Let us demonstrate through a basic example how the reservation price $\psi$ can be derived from a more comprehensive parameter space. This basic illustration presented in Example \ref{example1} is expanded in Appendix \ref{sec:multiperiod}, where we develop a multi-period framework covering an individual's entire lifespan,  demonstrating that our seemingly simple reduced-form framework is actually quite versatile and broadly applicable.

	\begin{example}\label{example1}
		
		If the reference period is short enough (e.g. one year) we may assume that the only uncertainty faced by an agent comes from the health status and drug efficacy, while the wealth status is known. Consider a specific agent labeled agent $a$, with wealth $w_a\in \mathbb R_+$ and assume that the agent attributes the probability $p_a$ to becoming sick and the probability $q_a$ to the success of the treatment if sick: $p_a$ and $q_a$ are subjective assessments that may differ from the real probabilities $r$ and $q$. Assume for simplicity that in case of disease and without treatment, agents lose the initial wealth entirely (as for instance in case of death), while if sick and treated they lose a proportion $\ell_a$ of their wealth---this may also account for the lost quality of life after recovery. The reservation price is then 
		\[
		\psi_a:=w_a q_a (1-\ell_a).
		\] 
		Hence, the reservation price may include subjective measures of therapeutic value $(q_a,\ell_a)$, that may be correlated to the actual one, personal personal fas well as personal financial factors. The expected wealth at time $0$ under choice $\texttt{O}$ is $W^{\texttt{O}}_a = w_a(1-p_a)$. 
		This ends our example.
	\end{example}

	As the reservation price $\psi_a$ represents the monetary value the agent assigns to receiving the treatment *if* diagnosed, then, at time $0$, the agent computes the expected value gain from the possibility of receiving treatment as $p_a\psi_a$. Denoting  $W^{\texttt{O}}_a$ the baseline expected wealth of agent $a$ under choice $\texttt{O}$ (take no action), that includes expected losses from illness, we can express the total expected wealth for agent $a$ (with parameters $(p_a, \psi_a)$) associated with each choice as follows:
	
	\begin{center}\small
		\begin{tabular}{c|l}
			\textit{Choice }& \textit{Expected wealth for agent $a$}\\
			\texttt{T}& $W^{\texttt{O}}_a + p_a \psi_a - p_a \theta$ \\
			\texttt{A}& $W^{\texttt{O}}_a + p_a \psi_a - \tilde\pi$\\
			\texttt{O} & $W^{\texttt{O}}_a$
		\end{tabular}
	\end{center}
	Note that $\theta$ and $\tilde\pi$ are market prices, that will be determined in equilibrium.

	We assume agents prefer a choice that leads to higher expected payoff under their distorted probabilities, and are indifferent between choices that lead to a same expected payoff. 
	Denote the (strict) preference relation of agent $a$ by $\succ_a$. 
	The preferences of agent $a$ over the set of choices $\{\texttt{T},\texttt{A},\texttt{O}\}$ are characterized as follows:
	\begin{itemize}
		\item $\texttt{O} \succ_a \texttt{T }$ $\Leftrightarrow$ $\psi_a < \theta$;
		\item $\texttt{A} \succ_a \texttt{O} $ $\Leftrightarrow$ $p_a \psi_a > \tilde\pi$;
		\item $\texttt{A} \succ_a \texttt{T} $ $\Leftrightarrow$ $p_a \theta> \tilde\pi$.
	\end{itemize}
	We emphasize that the conditions above are in line with a sequential decision process of the agents. Indeed, a decision of whether to buy insurance is taken at time 0, while a decision to buy or not the treatment is only faced conditionally on being diagnosed during the time interval [0,1] and by those agents that did not previously buy insurance. It can be checked that conditionally on being diagnosed, an uninsured agent $a$ that buys the treatment out of pocket has an expected future wealth of $\psi_a-\theta$, hence the agent strictly prefers to buy the treatment if $\psi_a>\theta$ (that is a deterministic condition). Hence, at time 0, an agent strictly prefers to buy insurance if $\tildepi<p_a(\theta\wedge \psi_a)$.

	Our two-dimensional parameter space provides a sufficient characterization of the population, which allows to determine preferences of agents over the set of choices $\{\texttt{T},\texttt{A},\texttt{O}\}$. 
	
	We can graphically represent the preferred choices of the agents by drawing (at most) three sub-regions of the parameter space, corresponding to the optimal choice among $\{\texttt{T},\texttt{A},\texttt{O}\}$, of an agent characterized by those parameters:
	\begin{itemize}
		\item Region $O(\theta,\tilde \pi) $: the subset of the parameter space representing agents for which choice $\texttt{O}$ is weakly preferred to choices $\texttt{A}$ and $\texttt{T}$, that is:
		\[
		O(\theta,\tilde \pi)=\{(p, \psi)\in[0,1]\times \mathbb R_+: \psi\le\theta, \; p\psi \le\tilde\pi\}
		\]Hence, for any agent $a$: 
		$(p_a,\psi_a)\in O(\theta,\tilde\pi)\Leftrightarrow (\texttt{O}\succeq_a \texttt{A}\text{ and }\texttt{O}\succeq_a \texttt{T}).$
		\item Region $A(\theta,\tilde \pi)$: the subset of the parameter space representing agents for which choice $\texttt{A}$ is strictly preferred to $\texttt{O}$ and $\texttt{T}$, that is
		\[
		A(\theta,\tilde \pi)=\{(p, \psi)\in[0,1]\times \mathbb R_+: p\theta>\tilde\pi, \; p\psi >\tilde\pi\}
		\]Hence, for any agent $a$: 
		$(p_a,\psi_a)\in A(\theta,\tilde\pi)\Leftrightarrow (\texttt{A}\succ_a \texttt{O} \text{ and }\texttt{A}\succ_a \texttt{T}).$
		\item Region $T(\theta,\tilde \pi)$: the subset of the parameter space representing agents for which choice $\texttt{T}$ is preferred strictly to $\texttt{O}$ and weakly to $\texttt{A}$, that is
		\[
		T(\theta,\tildepi)=\{(p, \psi)\in[0,1]\times \mathbb R_+: \psi>\theta, \; p\theta \le\tilde\pi\}
		\]Hence, for any agent $a$: 
		$(p_a,\psi_a)\in T(\theta,\tilde\pi)\Leftrightarrow (\texttt{T}\succ_a \texttt{O}\text{ and }\texttt{T}\succeq_a \texttt{A}).$
	\end{itemize}

	\begin{remark}
		Implicitly, there is a convention for indifferent  choices.  For instance, an agent that strictly prefers buying an insurance rather than buying the treatment out of the pocket if sick, but is indifferent between buying insurance and choice $\texttt{O}$, is assumed to belong to the region $O(\theta, \tildepi)$. 
		For fixed $\tilde\pi$ and $\theta$, the set of agents who are indifferent between two choices is always a Lebesgue set of measure $0$. Therefore, whenever  $\mu$ has a density, this convention has no impact on the resulting Nash equilibria. 
	\end{remark}
	
	\begin{notation}\label{not1} The probability measure $\mu$ on $[0,1]\times\mathbb R_+$ denotes  the distribution of agents' characteristics in the plane $(p,\psi)$. Also, we denote: 
		\begin{align*}
			\muT(\theta,\tildepi)&= \mu(T(\theta,\tildepi)),\\
			\muA (\theta,\tildepi)&= \mu(A(\theta,\tildepi)), \\
			\muO (\theta,\tildepi)&= \mu(O(\theta,\tildepi)).
		\end{align*}
		The dependence on  $\theta$ and $\tildepi$  will be  sometimes skipped,  if there is no confusion.	
	\end{notation}
	Figure \ref{fig:decision-zones} shows in the plane $[0,1]\times\mathbb R_+$ the three regions, $\muT,\;\muA,\; \muO$, for $\theta=1$ and $\tildepi=0,1$, while Figure \ref{fig2} (right) quantifies $\muO$ for different ranges of the couple drug price--insurance price. The underlying  distribultion $\mu$ considered is a product of a beta and an exponential, is pictured in the left hand subfigure.
	\begin{figure}\label{fig1}
		\begin{center}
			\begin{tikzpicture}
			\begin{axis}[
		xmin = 0, xmax = 1,
		ymin = 0, ymax = 1.5,
		xtick distance = 0.1,
		ytick distance = 0.1,
		grid = both,
		minor tick num = 4,
		major grid style = {lightgray},
		minor grid style = {lightgray!25},
		width = 0.8\textwidth,
		height = 1.2\textwidth]
		
		\addplot [domain=0:1, samples=10, blue] ({0.1},{1+x});
		\addplot [domain=0.1:1, samples=50, blue] ({x},{0.1/x});
		\addplot [domain=0:0.1, samples=10, blue] ({x},{1});
\end{axis}
\end{tikzpicture}
		\end{center}
		\caption{In abscissa, the subjective probability $p_a$, and in ordinate $\psi_a$. Here $\theta = 1$ and $\tildepi = 0.1$. Bottom left, $O$, top left, $T$, and top right $A$. If $\tildepi>\theta$ $O=\emptyset$  and only the regions $T $ and $O$ remain.}
		\label{fig:decision-zones}
	\end{figure}
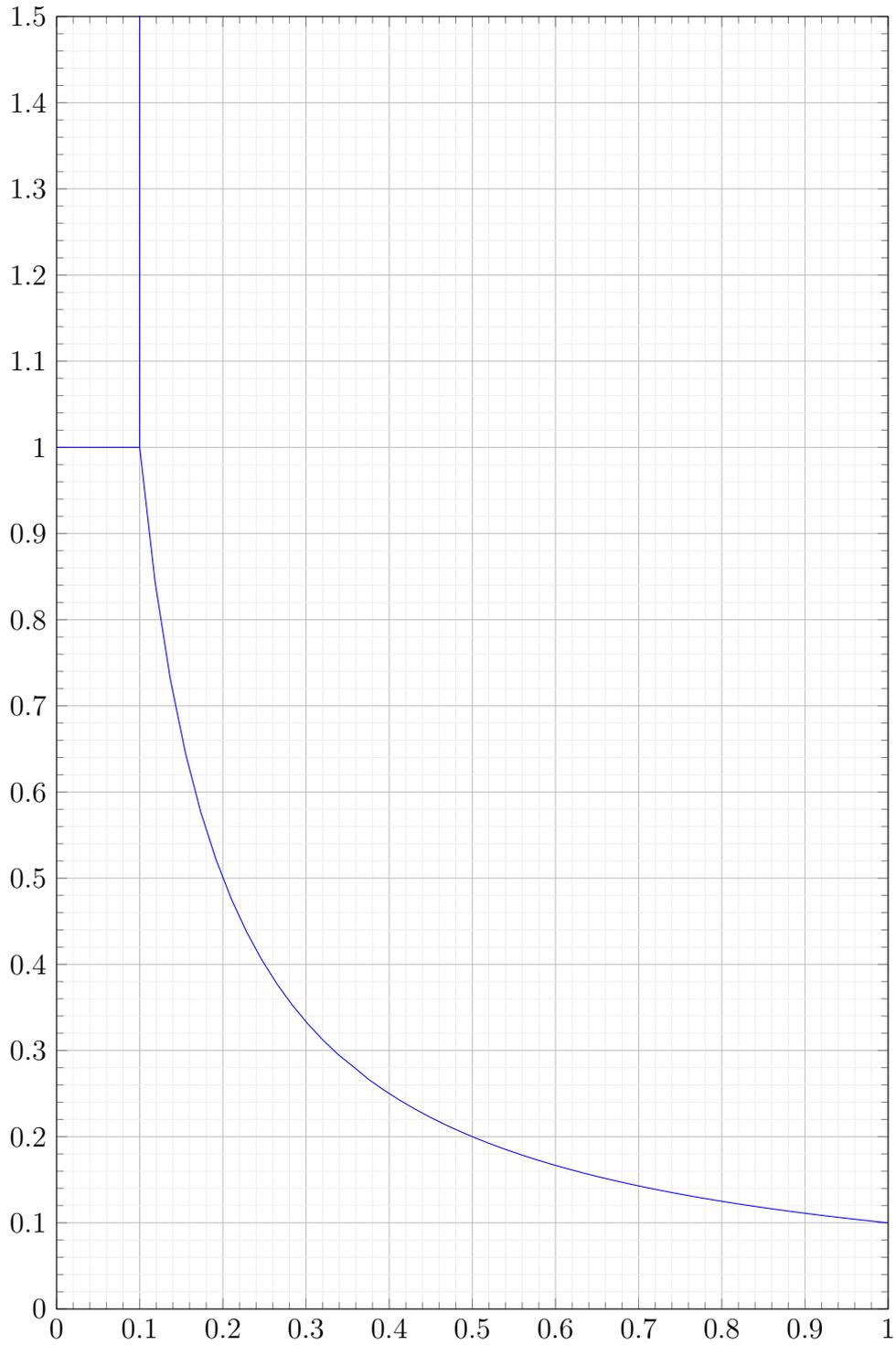\\
	\begin{figure}
		\begin{center}
			\includegraphics[scale=0.5]{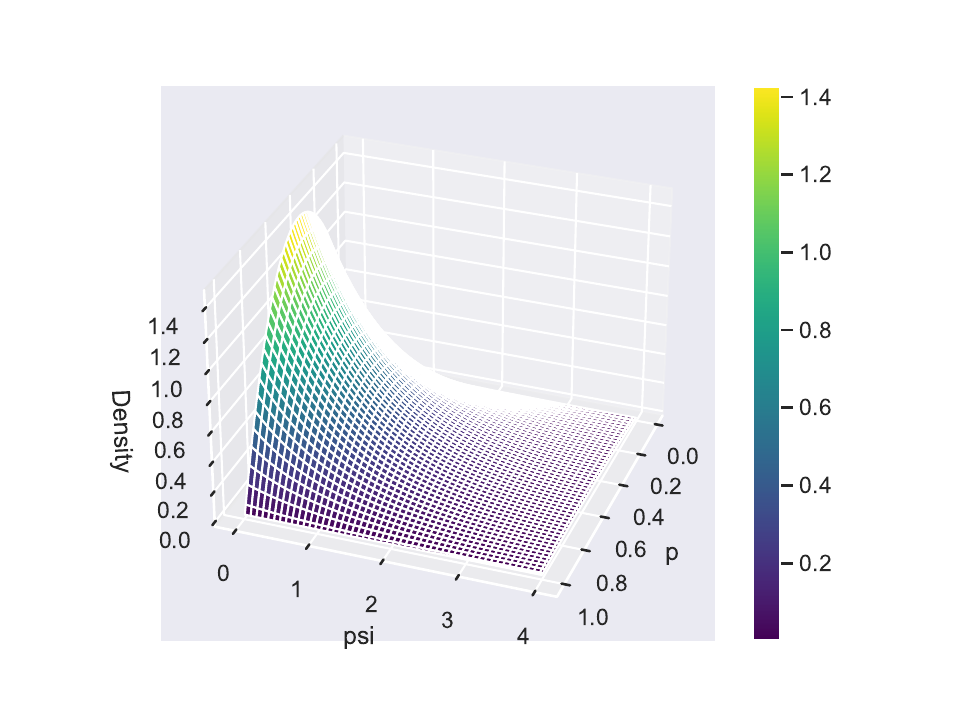}\includegraphics[scale=0.53]{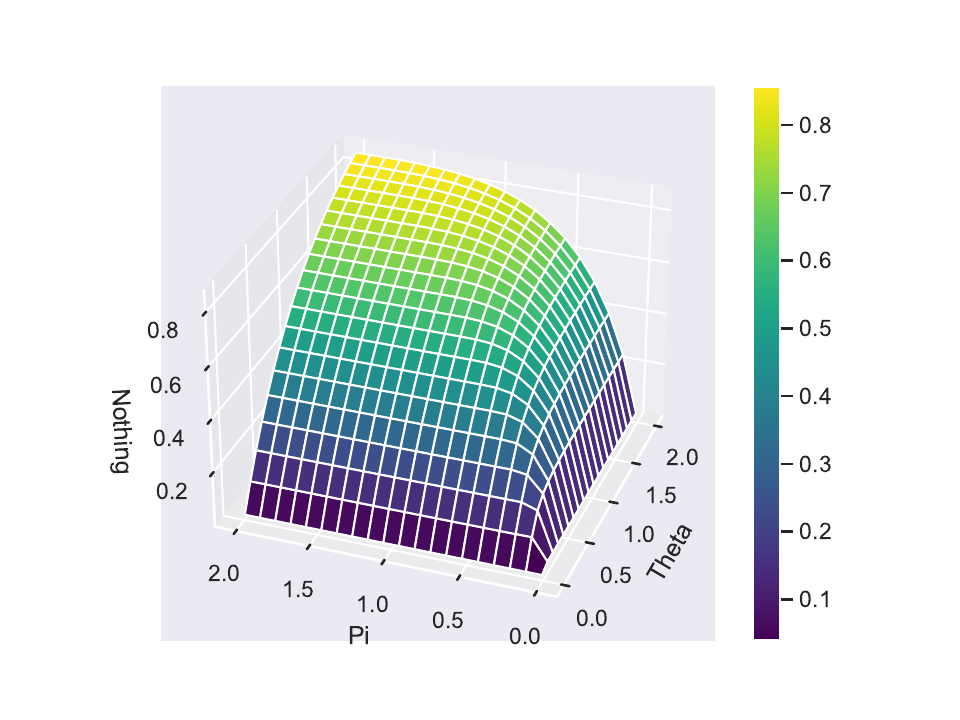}
		\end{center}
		\caption{Left:  population density function $\mu= \text{Beta}(2,2)\otimes\text{Exp}(1)$ in the space of characteristic $(p,\psi)$. Right: For this $\mu$,  the proportion of agents without drug access $\muO$,  for various price levels $(\theta, \pi)$.}
		\label{fig2}
	\end{figure}

	\subsection{The producer-insurer game}\label{sect:pharma-company-insurance}
	
	The game, played a time 0, is as follows. First, the producer chooses an action $\theta\in\mathbb R_+$,  interpreted as the price to be paid for obtaining the treatment. Second, the insurer reacts with a price $\tildepi\in\mathbb R_+$,  interpreted as the premium to be paid by an agent at time $0$, for an insurance covering the cost of the treatment in case of disease. Third, the population reacts to the price-premium couple $(\theta,\tildepi)$, as discussed in Subsection \ref{sect:Customers}. 
	Thus, the set of strategies for the producer is $\mathbb R_+$, and the set of strategies of the insurer is $\R_+^{\R_+}$, i.e., the set of \emph{functions} from $\R_+$ to $\R_+$. A couple $(\theta,\tildepi)\in\R_+^2$ is an action vector and a couple $(\theta,\pi)\in\R_+\times \R_+^{\R_+} $ is a strategy vector. We now define the \emph{payoffs} for the two players.
	
	\begin{definition}Let $\mathcal S:=\R_+\times \R_+^{\R_+}$ be the set of all strategy vectors of the game. Using Notation \ref{not1}, we define the payoff vector of the game  as the map $\mathcal S\to \R^2$, $(\theta,\pi)\mapsto (P_p(\theta,\pi(\theta)), P_i(\theta,\pi(\theta)))$, where:
		\begin{equation*}
			\Ppharma(\theta,\tildepi):=\theta r\left(\muA(\theta,\tildepi) + \muT(\theta,\tildepi)\right)
		\end{equation*} 
		represents the producer profit function, while  the insurer profit function is given as:
		\begin{equation*}
			\Pinsur(\theta,\tilde\pi):=(\tildepi-\theta r)\muA (\theta,\tildepi).
		\end{equation*}
	\end{definition}
	
	\begin{remark}[Zero cost assumption]
		Note that these profit formulas implicitly assume zero marginal costs for manufacturing, distribution, marketing, etc., for both the producer and the insurer. Incorporating positive costs is generally straightforward. For example, one could introduce marginal costs (say, $c_p \ge 0$ for the producer per treatment sold and $c_i \ge 0$ for the insurer per policy sold), leading to profit functions like $\Ppharma(\theta,\tildepi) = (\theta - c_p) r (\muA + \muT)$ and $\Pinsur(\theta,\tildepi) = (\tildepi - \theta r - c_i) \muA$. However, other cost structures (e.g., including fixed costs or costs dependent on factors other than volume) could also be modeled depending on the specific context. While such adjustments would change the specific quantitative results and potentially the details of the theorems presented later, they would not alter the fundamental structure of the game-theoretic model itself. For simplicity in this exposition, we proceed with the zero-cost assumption.
	\end{remark}
	
	
	\begin{definition}[Nash equilibrium]\label{def:Nash equilibrium}
		A pair of strategies $(\theta,\pi)\in\mathcal S$ is called a \emph{Nash equilibrium} if and only if we have 
		\begin{equation*}
			\Ppharma(\theta',\pi(\theta))\le \Ppharma(\theta,\pi(\theta)),\text{ for all }\theta'\in\R_+
		\end{equation*}
		and 
		\begin{equation}\label{eq:optimality-insurance}
			\Pinsur(\theta,\pi'(\theta))\le \Pinsur(\theta,\pi(\theta)),\text{ for all }\pi'\in\R_+^{\R_+} .
		\end{equation}
	\end{definition}
	\begin{remark}
		An equivalent condition to \eqref{eq:optimality-insurance} is
		$
		\Pinsur(\theta,\tildepi) \le\Pinsur(\theta,\pi(\theta))\text{ for all }\tildepi\in\R_+ .
		$
	\end{remark}
	
	\begin{remark}[The order of players matters]\label{rem:ordered-game}
		Consider an alternative game where both the producer and the insurer announce a price, without observing the move of the other player: the set of strategy vectors is then $(\R_+)^2$. In this case, a \emph{Nash equilibrium} is a pair of strategies $(\theta,\tildepi)\in\R_+\times\R_+$ such that $\Pinsur(\theta,\tildepi)\le \Pinsur(\theta, \pi)$ for all $\tildepi\in\R_+$, and $\Ppharma(\tilde\theta, \pi)\le \Ppharma(\theta,\pi)$ for all $\tilde\theta\in\R_+$. But for $\tildepi$ fixed, if $r>0$ and as long as $\muA(\theta_0,\pi)>0$ for some $\theta_0\in\R_+$, the producer would always be incentivized to let $\theta$ go to $\infty$ (note that $\muA(\theta,\tildepi)$ is increasing in $\theta$), so that there is no Nash equilibrium in this case. Alternatively, if $\tildepi$ is such that $\muA(\theta_0,\tildepi)=0$ for all $\theta_0$, then there are no agents that buy insurance, independently of the premium. In conclusion, the alternative game where neither the insurance nor the producer can react to the actions of the other, leads to either a trivial game (no equilibria) or a degenerate game in which the insurer is not a player.
	\end{remark}

	\begin{definition}[Subgame-perfect equilibrium]
		\label{def:subgame perfectness}
		A Nash equilibrium $(\theta,\pi)$ is called \emph{subgame-perfect} if
		\begin{equation}\label{eq:subgame perfect}
			\pi(\theta)\in\argmax_{\tildepi\in\R_+} \Pinsur(\theta,\tildepi)\quad \text{for \emph{all} } \theta\in\R_+.
		\end{equation}
	\end{definition}
	In Figure \ref{fig:3D_plot_colored} we plot in the plane $(p,\psi)$ the  population reaction, corresponding to a sub-game perfect Nash equilibrium. The  population is assumed to be distributed according  to the product of a $\mu= \mathrm{Beta}(2,2)\otimes \mathrm{Exp}(1)$, same as the one depicted Figure \ref{fig2} (right). 
	
	\begin{figure}
		\centering
		\includegraphics[scale=0.3]{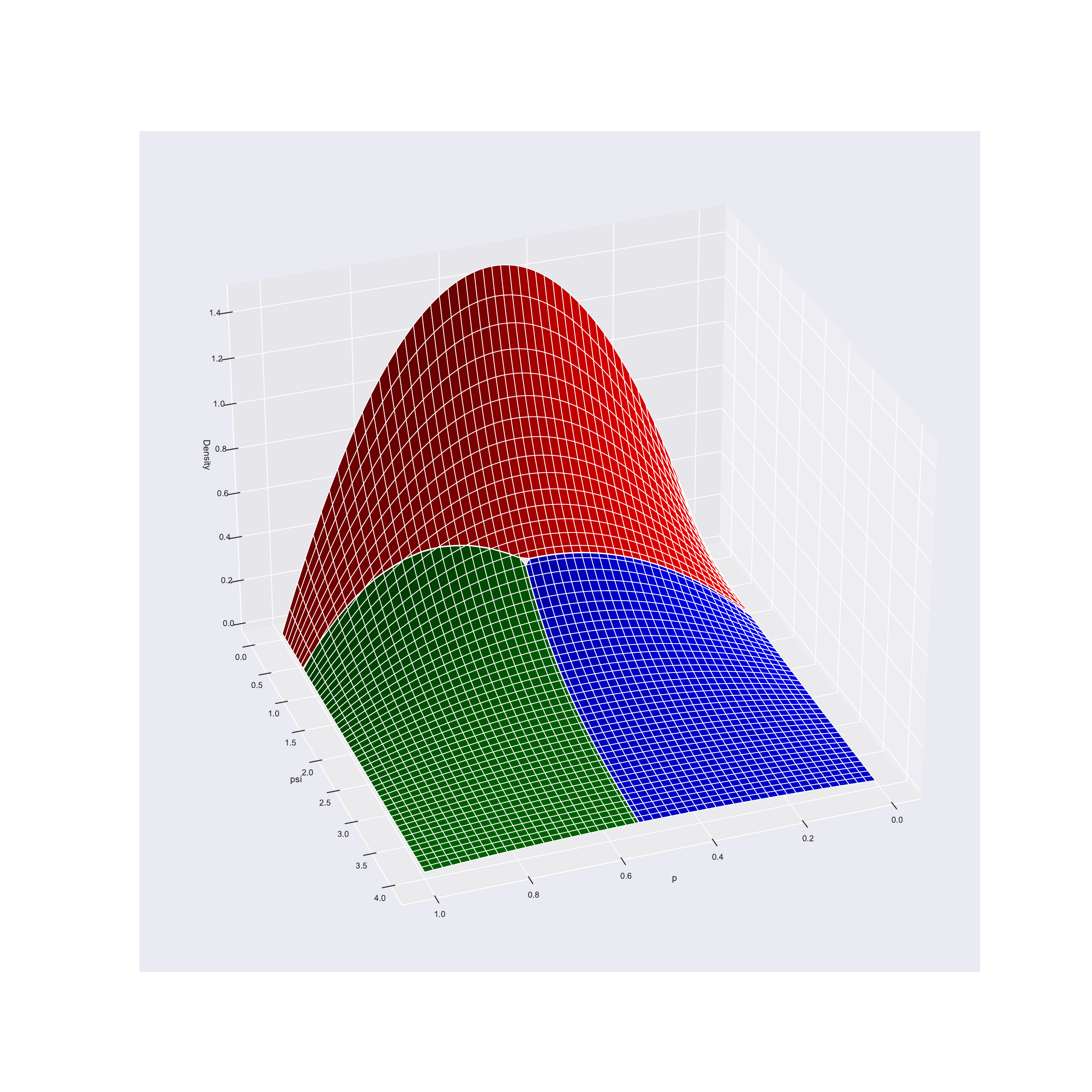}
		\caption{Subgame-perfect Nash equilibrium for $\mu= \mathrm{Beta}(2,2)\otimes \mathrm{Exp}(1)$.  The green region represents $\muA^*$ (the insured population), the blue region represents $\muT^*$ (uninsured agents that will buy the treatments out-of pocked upon diagnosis), and the red region represents $\muO^*$, the population without access to treatment.
		}
		\label{fig:3D_plot_colored}
	\end{figure}
	
	\begin{remark}[Subgame-imperfect and dictatorial Nash equilibria]
		Not all Nash equilibria in our model are necessarily subgame-perfect. A common example of games admitting non-subgame-perfect equilibria is the ultimatum game \cite{ultimatumGame}. Our model can also exhibit such equilibria (see also Section~\ref{sect:no-uniqueness-in-general}). Of particular interest is a specific type where the insurer leverages its potential to increase the producer's profit compared to a situation without an insurer. If a specific price pair $(\theta_0, \tildepi_0)$ yields higher profit for the producer than any scenario where the insurer doesn't enter the market, the insurer can issue an ultimatum: "Set the drug price to $\theta_0$, and I will set my premium to $\tildepi_0$; otherwise, I will not enter the market (i.e., set $\pi(\theta) = +\infty$ for $\theta \neq \theta_0$)." If the producer maximizes profit, they will comply and set the price to $\theta_0$.
		
		We can interpret such scenarios as \emph{dictatorial Nash equilibria}, where the insurer effectively dictates the price $\theta_0$ to the producer by threatening market exit. Thus, by doing so, the insurer effectively decides both the drug price and the premium.
	\end{remark}
	
	\begin{definition}[Dictatorial equilibrium]\label{def:dictatorial}
		A \emph{dictatorial equilibrium} is a Nash equilibrium $(\theta_0, \pi)$ where the insurer's strategy takes the form:
		\begin{itemize}
			\item $\pi(\theta_0) = \tildepi_0$ for some $\tildepi_0 \in [\theta_0 r, \theta_0)$,
			\item $\pi(\theta) = +\infty$ for all $\theta \neq \theta_0$.
		\end{itemize}
		In such an equilibrium, the insurer effectively dictates the price $\theta_0$ by refusing to enter the market for any other producer price.
	\end{definition}
	
	\begin{remark}[Maximizing insurer profit via dictatorial equilibria]\label{rem:maximizes-insurance-profit}
		If $(\theta^*, \pi^*)$ is a Nash equilibrium that maximizes the insurer's profit $\Pinsur(\theta^*, \pi^*(\theta^*))$ among all Nash equilibria, then a dictatorial equilibrium $(\theta^*, \pi_{dict})$ exists where $\pi_{dict}(\theta^*) = \pi^*(\theta^*)$ and $\pi_{dict}(\theta) = +\infty$ for $\theta \neq \theta^*$. This is because, according to Theorem~\ref{thm:existence-Nash-equilibrium-with-insurer}, in any Nash equilibrium $(\theta^*, \pi^*)$, the producer's profit $\Ppharma(\theta^*, \pi^*(\theta^*))$ is greater than or equal to the profit they could achieve if the insurer did not enter the market ($\sup_{\theta'} \Ppharma(\theta', +\infty)$). Therefore, the producer is incentivized to accept the dictated price $\theta^*$ under the ultimatum strategy $\pi_{dict}$.
	\end{remark}

	\section{Nash equilibria: existence, properties, and examples}\label{sect:existance-uniqueness}
	In this section, we first discuss existence and uniqueness of Nash equilibria for the model presented in Section \ref{sect:math-model}. Then, we provide examples to compare the outcome, in terms of profits and overall healthcare coverage in the population, for the cases where an insurance is present and where no insurance is present.

	\begin{assumptions}In the remainder, we assume that:
		\begin{itemize}
			\item[1.] the incidence rate $r$ is contained in the open interval $(0,1)$;
			\item[2.] there exists a potential for profit for the producer: there exist $(\theta,\tildepi)\in[0,1]\times\R_+$ with $\Ppharma(\theta,\tildepi)>0$ (which is equivalent to: $\mu([0,1]\times \{0\})<1$) .\end{itemize} 
	\end{assumptions}
	We start by establishing a result of existence. The second half of this section contains examples where no Nash equilibrium exists. 
	
	\subsection{Existence of Nash equilibria} 
	
	Generally, we can neither expect existence nor uniqueness of Nash equilibria for our model. Weak assumptions ensure existence of general and of subgame-perfect equilibria, see Theorem \ref{thm:existence-Nash-equilibrium} below, but uniqueness is still not guaranteed. 
	
	Some strategies of the insurer (e.g. when the profit is negative) are strongly dominated independently of the distribution $\mu$ of the population. Let $\theta$ be fixed; we notice that any premium $\tildepi< \theta r$ leads to negative profits for the insurer while any premium $\tildepi\geq\theta $ leads to $A(\theta,\tildepi)=\emptyset $ (no agents buy insurance) hence null profit for the insurer. 
	
	We thus denote
	\begin{equation}\label{eq:Delta}
		\Delta:= \{(\theta,\tildepi)\;|\;\theta \geq 0, \tildepi \in [\theta r,\theta]\}
	\end{equation} 
	the set of $(\theta, \tildepi)$ that may lead to a positive profit for the insurer, and we use the convention
	$\tildepi=+\infty \Rightarrow \tildepi \in[\theta,+\infty)$, that is, $\{\tildepi=+\infty\}$ denotes any action of the insurer leading to no agents being insured. Hence, the action $\tildepi=+\infty$ is interpreted as \emph{the insurer does not enter the market}.

	We can focus our attention on actions $(\theta,\tilde\pi)\in\Delta$, where the insurance profit is non-negative, as well as the case $\tilde\pi=\infty$, where the insurer refuses to enter the market.
	
	\begin{theorem}[Existence of Nash equilibria]\label{thm:existence-Nash-equilibrium}
		Assume that $\Ppharma$ and $\Pinsur$ are continuous, and that
		\begin{equation}\label{eq:pharma-profit-to-0}
			\lim_{\theta\to\infty} P_p(\theta,+\infty) = 0.
		\end{equation}
		Then there exists at least one subgame-perfect Nash equilibrium as defined in Definition \ref{def:subgame perfectness}. In particular, there exists at least one Nash equilibrium as in Definition \ref{def:Nash equilibrium}. 
		
		Furthermore, there exists a subgame-perfect Nash equilibrium $(\theta,\pi)$ that maximizes the profit of the producer, i.e. for any Nash equilibrium $(\theta',\pi')$, we have
		\begin{equation*}
			\Ppharma(\theta,\pi(\theta))\ge\Ppharma(\theta',\pi'(\theta')).
		\end{equation*}
	\end{theorem}

	We postpone the proof of Theorem \ref{thm:existence-Nash-equilibrium} to the end of this section. 
	\begin{remark}
		If $\sup \Ppharma(\cdot, +\infty) = +\infty$, then there is no Nash equilibrium, because the producer can always change $\theta$ to increase its profits. The condition \eqref{eq:pharma-profit-to-0} is sufficient to avoid this case. See Section \ref{sect:no-Nash-equilibria} for other examples that break the assumptions of Theorem \ref{thm:existence-Nash-equilibrium} and where there is no Nash equilibrium. 
	\end{remark}
	
	Note the useful inequalities, valid for every $(\theta,\tildepi) \in \bar \Delta$:
	\begin{align}
		\label{eq:bound-Ppharma}
		& \Ppharma(\theta,+\infty) \leq \Ppharma(\theta,\tildepi) \leq r^{-1} \Ppharma(\theta r,+\infty) ,\\
		\label{eq:Ppharma-without-insurer-equals-tail}
		& \Ppharma(\theta,+\infty) \leq r \int_{[0,1]\times[\theta,\infty)} \psi\,\mathrm \d\mu(p, \psi) .
	\end{align}
	The first one is established by noting that $O(\cdot, \cdot)$ is increasing in both coordinates, therefore, $\tildepi \in [\theta r, \theta]$ imposes that 
	$
	O(\theta r, \theta r) \subset O(\theta r,\tildepi) \subset O(\theta,\tildepi) \subset O(\theta,\theta).
	$
	As for the second inequality, 
	$$(\mathcal{A}+\mathcal T)(\theta,+\infty)=\mathcal T(\theta,+\infty)=\mu([0,1]\times[\theta,\infty))
	$$ hence 
	\begin{equation*}
		\Ppharma(\theta,+\infty) = r\int_{[0,1]\times[\theta,\infty)} \theta\,\d\mu(p,\psi) \leq r\int_{[0,1]\times[\theta,\infty)} \psi\,\mathrm\d\mu(p,\psi).
	\end{equation*}
	
	\begin{remark}[Relation of Theorem \ref{thm:existence-Nash-equilibrium} to the literature on existence of Nash equilibria]
		The classical \emph{Debreu-Fan-Glicksberg (DFG) Theorem} (see \cite{futi91}), stemming from the contributions \cite{deb52, fan52, gli52}, states that a non-cooperative perfect-information $n$-person game with compact action spaces and appropriate continuity and quasi-convexity assumptions on the payoff functions possesses an equilibrium in pure strategies.
		Comparing with our Theorem \ref{thm:existence-Nash-equilibrium}, the assumption \eqref{eq:pharma-profit-to-0} broadly speaking ensures that $\theta$ is chosen inside a compact set (if $\theta$ is too large, $\Ppharma$ is small), while continuity of the payoffs is a direct assumption of Theorem \ref{thm:existence-Nash-equilibrium}. However, we do not expect the quasi-convexity assumptions needed for the DFG Theorem to hold for our model. Therefore, we give a direct of proof of Theorem \ref{thm:existence-Nash-equilibrium} below.
	\end{remark}
	It is useful to also emphasize some characteristics of the population of agents that ensure existence of the Nash equilibrium:
	\begin{lemma}[Stronger condition for existence of Nash equilibrium]\label{cor:existence-Nash-equilibrium}
		The assumptions of Theorem \ref{thm:existence-Nash-equilibrium} are satisfied if 
		\begin{equation}\label{eq:finite-psi-moment}
			\int_{[0,1]\times\R_+} \psi\, \d\mu(p, \psi) < \infty,
		\end{equation}
		and $\mu$ has a density with respect to the Lebesgue measure.
	\end{lemma}
	
	\begin{proof}The condition \eqref{eq:finite-psi-moment} implies that $\lim_{\theta\to\infty}\Ppharma(\theta,+\infty) = 0$ by \eqref{eq:Ppharma-without-insurer-equals-tail} and the Monotone Convergence Theorem. 
		
		Continuity of $\Ppharma$ and $\Pinsur$ follows if we can prove the continuity of $\muA$, $\muO$ and $\muT$. To do so, we notice that for every sequence $((\theta_n,\tildepi_n))_{n\in\mathbb N}$ converging to $(\theta,\tildepi)$ as $n\to\infty$, we have
		\begin{equation}\label{eq:sandwiching-liminf-limsup-A}
			\mu(\mathring A(\theta,\tildepi)) \leq \liminf_{n\to\infty} \muA(\theta_n,\tildepi_n) \leq \limsup_{n\to\infty} \muA(\theta_n,\tildepi_n) \leq \mu(\bar A(\theta,\tildepi)),
		\end{equation}
		where $\mathring A$ is the interior of $A$ and $\bar A$ is its closure. But $\bar A \setminus \mathring A$ is a Lebesgue null set by the definition of the set $A$. Thus, by absolute continuity of $\mu$ with respect to the Lebesgue measure, we also have $\mu(\bar A \setminus \mathring A) = 0$. Consequently, $\mu(\mathring A(\theta,\tildepi))=\mu(\bar A(\theta,\tildepi))$, which together with \eqref{eq:sandwiching-liminf-limsup-A} establishes continuity of $\muA$. The same holds for $\muO$ and $\muT$ . \end{proof}
	The following Lemmata are used in the proof of Theorem \ref{thm:existence-Nash-equilibrium}.
	\begin{lemma}[Compactness of super-level sets]
		\label{lemma:compacity-profits}
		Under the assumptions of Theorem \ref{thm:existence-Nash-equilibrium}, for every $t>0$, the following sets are compact:
		\begin{equation}\label{eq:super level pharma}
			\{(\theta,\tildepi)\in\Delta:\Ppharma(\theta,\tildepi) \geq t\}
		\end{equation}
		and 
		\begin{equation}\label{eq:super level insurance}
			\{(\theta,\tildepi)\in\Delta:\Pinsur(\theta,\tildepi) \geq t\}.
		\end{equation}
	\end{lemma}
	
	\begin{proof}
		The sets are closed because $\Ppharma$ and $\Pinsur$ are continuous. To show that \eqref{eq:super level pharma} is bounded, use \eqref{eq:bound-Ppharma} and \eqref{eq:Ppharma-without-insurer-equals-tail}:
		$$ \Ppharma(\theta,\tildepi) \leq r^{-1}\Ppharma(\theta r,+\infty) \to 0 \quad \text{ as } \theta\to\infty . $$
		It thus follows that for every $t>0$, there must be a $M<\infty$ such that $\Ppharma(\theta,\tildepi)\geq t$ implies $\theta\leq M$ (and thus $\tildepi\leq M$). We now show that \eqref{eq:super level insurance} is bounded. Note that the point $(\theta,+\infty)$ does not belong to the set \eqref{eq:super level insurance} for $t>0$. Hence, for $\tildepi\in[r\theta,\theta)$,
		\begin{equation*}
			\Pinsur(\theta,\tildepi) = (\tildepi-\theta r) \muA(\theta,\tildepi) \leq r^{-1} \frac{\tildepi}\theta \Ppharma(\theta,\tildepi) \leq r^{-1} \Ppharma(\theta,\tildepi).
		\end{equation*}
		Thus $\Pinsur(\theta,\tildepi) \geq t$ implies that $(\theta,\tildepi)$ is contained in \eqref{eq:super level pharma}, and the latter set is bounded. 
	\end{proof}
	
	\begin{lemma}[Properties of insurer's profit-maximizing actions]
		\label{lemma:compactness-area-Nash-equilibria}
		Define $K$ as the set of all actions $(\theta,\tildepi)\in \Delta$ with $\Pinsur(\theta,\tildepi) = \sup \Pinsur(\theta,\cdot)$. Under the assumptions of Theorem~\ref{thm:existence-Nash-equilibrium}, 
		\begin{enumerate}
			\item $K$ is closed, 
			\item for every $\theta \in \R_+$, there exists a $\tildepi$ such that $(\theta,\tildepi) \in K$ (possibly $\tildepi=\theta$), 
			\item $\Ppharma$ attains its maximum on $K$, and there exists $M<\infty$ such that all the maximizers $(\theta,\tildepi)\in\Delta$ of $\Ppharma$ satisfy $\theta\leq M$. \label{item:last} 
		\end{enumerate}
	\end{lemma}

	\begin{proof}
		To prove that $K$ is closed, consider a sequence $((\theta_n,\tildepi_n))_{n\in\mathbb Z_{\ge 0}}$ in $K$ that converges to some $(\theta,\tildepi)\in\Delta$. Then $\Pinsur(\theta_n,\tildepi_n) \to \Pinsur(\theta,\tildepi)$ as $n\to\infty$. On the other hand, $ \sup \Pinsur(\theta_n,\cdot) \to \sup \Pinsur(\theta,\cdot)$ as $n\to\infty$. We conclude that $(\theta,\tildepi)\in K$. 
		
		We now prove the second claim: Fix $\theta\in\R_+$. By continuity of $\Pinsur$ on the compact set $[\theta r, \theta]$, the maximum of $\Pinsur$ is attained, thus $K \cap (\{\theta\}\times[\theta r, \theta])$ is nonempty. 
		Finally, fix $t = \sup \Ppharma(\cdot, +\infty)$. The set \eqref{eq:super level pharma} is compact by Lemma \ref{lemma:compacity-profits}, thus $\Ppharma$ reaches its maximum on it; as $\Ppharma$ is strictly smaller than this maximum outside of \eqref{eq:super level pharma}, the conclusion follows. 
	\end{proof}
	Having established the previous Lemmata, we can now prove Theorem \ref{thm:existence-Nash-equilibrium}.
	\begin{proof}[Proof of Theorem \ref{thm:existence-Nash-equilibrium}]
		Let $(\theta,\tildepi) \in K$ be a maximizer of $\Ppharma$ on $K$. Then every $(\theta,\pi)$, where $\pi \in \R_+^{\R_+}$ is such that $\pi(\theta) = \tildepi$ and $(\theta',\pi(\theta')) \in K$ for every $\theta'\in\R_+\setminus\set\theta$, is a subgame-perfect Nash equilibrium. 
		
		We now prove the last assertion in Theorem \ref{thm:existence-Nash-equilibrium}, namely that there exists a subgame-perfect Nash equilibrium $(\theta,\pi)$ which maximizes the profit of the producer among all Nash equilibria.
		Since $\set{\tildepi\in[\theta r,\theta]: (\theta',\tildepi)\in K}$ is compact for every $\theta'\in\R_+$ by Lemma \ref{lemma:compactness-area-Nash-equilibria}, and since $\Ppharma$ is continuous by assumption, we can choose $\pi$ such that, for each $\theta'\in\R_+\setminus\set\theta$, we have $(\theta',\pi(\theta'))\in K$ and 
		\begin{equation}\label{eq:locally-dominating}
			\Ppharma(\theta', \pi(\theta')) \ge \Ppharma(\theta', \tildepi)
		\end{equation}
		for every $\tildepi\in[\theta' r,\theta']$ for which $(\theta', \tildepi)\in K$. Since $K$ is compact, by \ref{item:last} of Lemma~\ref{lemma:compactness-area-Nash-equilibria} there exists a $\theta\in\R_+$ such that $(\theta,\pi(\theta))\in K$ and such that
		\begin{equation}\label{eq:dominating-theta}
			\Ppharma(\theta,\pi(\theta))\ge\Ppharma(\theta',\pi(\theta'))
		\end{equation}
		for all $\theta'\in\R_+$. The tuple $(\theta,\pi)$ is a Nash equilibrium by the first half of this proof, and furthermore by construction it maximizes the profit of the producer among all Nash equilibria: Indeed, by \eqref{eq:dominating-theta} and \eqref{eq:locally-dominating},
		$
		\Ppharma(\theta,\pi(\theta))\ge\Ppharma(\theta',\pi(\theta'))\ge \Ppharma(\theta',\pi'(\theta'))
		$
		for any other Nash equilibrium $(\theta',\pi')$.
	\end{proof}

	\subsection{Examples where no Nash equilibrium exists}\label{sect:no-Nash-equilibria}
	
	If either of the conditions of Theorem \ref{thm:existence-Nash-equilibrium} fails, there may not exist any Nash equilibria. Below are two illustrations. They are not chosen for realism from a practical application standpoint, but rather to shed light on the importance of the assumptions in Theorem \ref{thm:existence-Nash-equilibrium}. 
	
	\begin{example}[Lack of regularity]
		Consider the Dirac measure $\mu=\delta_{(x,0)}$ for some $x>0$, meaning that all individual are identical and for simplicity they do not believe they can get sick (so nobody buys insurance). The incidence of the disease  is nevertheless contained in $(0,1)$. In this case, no Nash equilibrium can exist. Indeed, for any $\theta<x$, the producer can increase its profit by increasing $\theta$ while staying below $x$. However, for any $\theta\ge x$, the producer can increase its profit by choosing any $\theta$ below $x$.
	\end{example} 
	Note that the example above relies on the ambiguity on the best decision by the agents in the population, when two options are equivalent, i.e. when the individual lies on the boundary between several regions. For example, if we change the definition of $\texttt{T}$ to include the boundary between $\texttt{T}$ and $\texttt{O}$ (meaning that, when buying the treatment has the same cost as not buying the treatment, the individual chooses to buy the treatment), then the above example does, in fact, have a Nash equilibrium.
	Fortunately, the issue is avoided if $\mu$ has a density with respect to the Lebesgue measure, as seen in Remark \ref{cor:existence-Nash-equilibrium}.

	\begin{example}[Lack of integrability]
		Consider $\mu = \delta_p\otimes\mu_{\psi}$ for some $p\in (0,r)$, where $\delta_p$ denotes the Dirac measure concentrated on $p$ and $\mu_{\psi}$ is a measure on $\R_+$ defined by
		
		\begin{equation*}
			\mu_{\psi}([0,t]) = \max\left(0, 1 - t^{-\alpha}\right),\text{ for all }t>o\text{ and for some $\alpha\in(0,1)$.}
		\end{equation*}
		Then for every $\theta\geq 1$, the profit of the producer is $\theta^{1-\alpha}$, which increases to infinity as $\theta\to\infty$. There is thus no maximum.
	\end{example}
	This example is however unrealistic. Indeed, the average price that an individual is willing to pay, given by 
	$
	\int_{[0,1]\times\R_+} \psi\, \d\mu(p, \psi)$,
	is infinite. Alternatively and equivalently, the earnings of the producer can be made arbitrarily large. This does not reflect real-life situations, where the amount of wealth in the world is finite.

	\subsection{No uniqueness in general}\label{sect:no-uniqueness-in-general}
	Unlike non-existence, non-uniqueness seems to be a generic property rather than a pathological one, as illustrated in the examples below.
	
	\begin{example}
		Consider the population is characterized by the measure $\mu = \eta \otimes\mu_{\psi}$, where $\eta$ is the uniform measure on $[0,r/2]$ and $\mu_{\psi}$ is a measure on $\R_+$ defined by
		\begin{equation*}
			\mu_{\psi}([0,t]) = 
			\begin{cases}
				0 &\text{if } t<1 \\
				1-\frac{1}{t} &\text{if }t\in [1,2) \\
				\frac{3+t}{6} &\text{if } t\in [2, 3] \\
				1 &\text{if } t>3 .
			\end{cases}
		\end{equation*}
		This measure satisfies the hypotheses of Theorem \ref{thm:existence-Nash-equilibrium}. The profit of the producer is, for any $\theta\in [1,2]$, given by $r \theta \mu_{\psi}((\theta,\infty)) = r$, and is strictly smaller for $\theta \notin [1,2]$. The strategy chosen by the insurer does not play a role, since the set $A$ has measure $0$ for any $(\theta,\tildepi)\in \Delta$ under $\mu$ (agents are over optimistic regarding the diagnosis chances and do not buy insurance). Hence \emph{any} $(\theta,\pi)$ for $\pi\in\R_+^{\R_+}$ such that $\theta\in [1,2]$ and $\pi(\theta')\geq \theta' r$ for every $\theta'$ is a Nash equilibrium. 
	\end{example}
	
	\begin{example}
		Take $\mu = \delta_{(p,x)}$ for some $p\in (r,1)$ and $x>0$. From the insurer's point of view,  for every $\theta<px/r$, the price of the insurance should be exactly $\pi = px$: for any price above no one would subscribe to insurance, while any price below yields suboptimal profits to the insurer. If $\theta\geq px/r$ however, the insurer does not make a profit, so they should refuse to enter the market. Now consider the producer's point of view. If the insurer does not enter the market, the producer must set $\theta = x$. On the other hand, if the insurer enters the market (and, by necessity, sets $\pi=px$), the producer should then choose the highest possible $\theta$ that still ensures a non-negative profit to the insurer, which is $px/r > x$. 
		From the above considerations, it is clear that for every $\theta_0 \in (x,px/r)$, the couple $(\pi(\cdot), \theta_0)$ where $\pi(\theta) = px$ for every $\theta \leq \theta_0$ and $+\infty$ otherwise is a Nash equilibrium. Intuitively speaking, this corresponds to the insurer telling the producer that they will not enter the market if the drug price exceeds $\theta_0$, in a situation where the producer will earn more by complying rather not (and having the insurer not enter the market). 
	\end{example}
	While the last example breaks the regularity assumption of Theorem \ref{thm:existence-Nash-equilibrium}, one could make $\mu$ arbitrarily smooth while ensuring $
	\int_{[0,1]\times\R_+}\psi\,\mathrm d\mu(p,\psi)<\infty ,$
	thus ensuring that the assumptions of Theorem \ref{thm:existence-Nash-equilibrium} are satisfied.

	\subsection{When does the insurer enter the market?}
	In some population distributions, no Nash equilibria sees the insurer enter the market. 
	
	\begin{example} 
		Fix $\mu = 0.99\delta_{(0,x)} + 0.01 \delta_{(1, 2x)}$, and $r=1/2$. The only reasonable price for the insurer is $\tildepi = 2x$. For the insurer to enter market, we need its profit to be non-negative, which implies $\theta \leq 2x/r = 4x$. The maximum profit that can be made by the producer is then attained for $\theta = 4x$ and equals $0.04 x$. On the other hand, the producer can make a profit of $0.99x$ by setting $\theta = x$. This is clearly larger, so any $(x,\pi)$ with $\pi\in[x,\infty)^{\R_+}$ is a Nash equilibrium, and in none of these equilibra does the insurer enter the market. 
	\end{example}
	
	This example relies on the existence of regions with zero $\mu$-mass, in this case the region $[0,1]\times[0,x)$. If none exists, then the following theorem guarantees that every Nash equilibrium sees the insurer enter the market. 
	\begin{theorem}[Profits of insurance and producer in Nash equilibria]
		\label{thm:existence-Nash-equilibrium-with-insurer}
		Assume the conditions of Theorem \ref{thm:existence-Nash-equilibrium} hold. Then, in every Nash equilibrium $(\theta,\pi)$, the following hold:
		\begin{enumerate}
			\item the producer makes more profit than if the insurer did not enter the market:
			\begin{equation}\label{eq:profit-pharma-increase}
				\Ppharma(\theta,\pi(\theta)) \geq \sup_{\theta'\in\R_+}\Ppharma(\theta',+\infty);
			\end{equation}
			\item if the Lebesgue measure is absolutely continuous with respect to $\mu$, the insurer makes a non-zero profit:
			\begin{equation*}
				\Pinsur(\theta,\pi(\theta)) > 0.
			\end{equation*}
		\end{enumerate}
	\end{theorem}
	
	\begin{remark}[Impact of the presence of an insurer on treatment price and access to treatment]
		We shall see in Section \ref{sec:example} that the presence of an insurer can lead to various outcomes in terms of treatment price and access to treatment. There are examples where the price increases, examples where the price decreases, examples where access to treatment improves, and examples where access to treatment worsens.
	\end{remark}
	Theorem~\ref{thm:existence-Nash-equilibrium-with-insurer} immediately follows from the following Lemma, which states that for every fixed treatment price $\theta>0$, the profit for the insurer attains a positive maximum at $\frac{\theta}{r} < \tildepi < \theta$, and the producer's profit is larger than it would be if the insurer refrains from entering the market. 
	
	\begin{lemma}
		Under the conditions of Theorem \ref{thm:existence-Nash-equilibrium-with-insurer}, for every fixed $\theta\in\R_+$, the maximum of $\Pinsur(\theta,\cdot)$ is realized at some $\tildepi \in (r\theta, \theta)$ and strictly positive, and $\Ppharma(\theta,\tildepi) > \Ppharma(\theta,+\infty)$. 
	\end{lemma}
	
	\begin{proof}
		Let $\theta\in\R_+$. The function $[\theta r, \theta] \to \R, \tildepi \mapsto \Pinsur(\theta,\tildepi)$ is continuous,  thus attains its maximum, and is zero for $\tildepi \in \{\theta r, \theta\}$. Taking any $\tildepi \in (\theta r, \theta)$ ensures that $\tildepi-\theta r >0$, and (by absolute continuity of the Lebesgue measure with respect to $\mu$) that $\muT(\theta,\tildepi)>0$, and thus that $\Pinsur(\theta,\tildepi)>0$ and $\sup \Pinsur(\cdot, \theta)>0$. 
		
		We now show that, for any $(\theta,\tildepi)\in\Delta$, if $\Pinsur(\theta,\tildepi)=\sup\Pinsur(\cdot,\theta)$, then $\Ppharma(\theta,\tildepi)>\Ppharma(\theta,+\infty)$: Indeed, if $\Pinsur(\theta,\tildepi)=\sup\Pinsur(\cdot,\theta)$, then because $\tildepi < \theta$, we have, denoting the Lebesgue measure on $\R^2$ by $\text{Leb}$, $\mathrm{Leb}(O(x,y)) < \mathrm{Leb}(O(x,+\infty))$, implying by absolute continuity $\muO(x, y) < \muO(x,+\infty)$, thus $\Ppharma(x, y) > \Ppharma(x,+\infty)$.
	\end{proof}

	To conclude the description of Nash equilibria, we give a characterization of equilibrium points:
	
	\begin{lemma}[Characterization of Nash equilibria]\label{lem:characterization of Nash equilibria}
		Recall $K$ as in Lemma \ref{lemma:compactness-area-Nash-equilibria}. 
		Write $K'$ for the set of all $(\theta,\tildepi) \in K$ such that $\Ppharma(\theta,\tildepi) \geq \sup \Ppharma(\cdot, +\infty)$. Then
		\begin{enumerate}
			\item every Nash equilibrium $(\theta,\pi)$ satisfies $(\theta,\pi(\theta))\in K'$,
			\item under the hypotheses of Theorem \ref{thm:existence-Nash-equilibrium}, for every $(\theta,\tildepi) \in K'$, there exists a Nash equilibrium $(\theta,\pi)$, by setting $\pi(\theta) = \tildepi$ and $\pi(z)=z$ for every $z\neq \theta$.
		\end{enumerate}
	\end{lemma}
	The proof of both statements is straightforward and omitted for brevity.

	\begin{remark}[Compactness of $K'$]
		The set $K'$ is compact, since it is the intersection of $K$ with 
		$
		\{(\theta,\tildepi)\in\Delta: \Ppharma(\theta,\tildepi) \geq \sup \Ppharma(\cdot, +\infty)\},
		$
		which is compact by Lemma \ref{lemma:compacity-profits}. Thus $\Pinsur$ attains its maximum on $K'$. This maximum is the maximum profit that the insurer can attain at a Nash equilibrium. 
	\end{remark}
	
	\begin{remark}[First attempt at using the Implicit Function Theorem: through derivatives]\label{rem:implicit-function-Theorem}
		Assume that $\Pinsur$ and $\Ppharma$ are sufficiently differentiable. Consider the set
		\begin{equation*}
			\tilde K = \set*{(\theta,\tildepi)\in\Delta: \frac{\partial\Pinsur}{\partial\tildepi}(\theta,\tildepi)=0, \frac{\partial^2\Pinsur}{\partial\tildepi^2}(\theta,\tildepi)<0}
		\end{equation*}
		of strict local maximizers in the first coordinate of $\Pinsur$. By the implicit function Theorem, for any $(\theta_0, \tildepi_0)\in\tilde K$, there exists a neighborhood $(\theta_0,\tildepi_0)$ such that one can write $(\theta,\tildepi)=(\theta,g(\theta))$ for some function $g$ and all $(\theta,\tildepi)\in\tilde K$ which are contained in this neighborhood. Therefore, $\tilde K$ is, locally, the image of a curve. 
		Unfortunately, recalling $K$ from Lemma \ref{lemma:compactness-area-Nash-equilibria}, we need not have $K\subset\tilde K$, since maximizers may occur also with the second derivative equal to $0$, nor $\tilde K\subset K$, since local maximizers may not be global maximizers. 
	\end{remark}
	\begin{remark}[Second attempt at using the implicit function Theorem: directly]\label{rem:implicit-function-Theorem-2}
		We present an alternative approach: Consider $K$ as in Lemma \ref{lemma:compactness-area-Nash-equilibria} and define $\tilde P: \Delta\to\R$ as $\tilde P(\theta,\tildepi) = \sup \Pinsur(\theta,\cdot)-\Pinsur(\theta,\tildepi)$. Then $K$ is the set of all $(\theta,\tildepi)\in\Delta$ such that $\tilde P(\theta,\tildepi) = 0$. Assume that $\tilde P$ is continuously differentiable (this assumption does not follow from differentiability of $\Pinsur$). Then 
		$
		\partial_{\tildepi}\tilde P = -\partial_{\tildepi}\Pinsur.
		$
		Therefore by the Implicit Function Theorem, at any point $(\theta,\tildepi)\in K$ for which $\partial_{\tildepi}\Pinsur(\theta,\tildepi)\neq 0$, the set $K$ is the image of a curve in a neighborhood of that point. 
	\end{remark}

	\subsection{Some example of notable effects of insurance on prices and coverage}\label{sec:example}
	In this section, we emphasize some non-trivial effects of adding an insurance on the price of treatments, and on the overall access to treatments. We show that in some example it leads to increased overall healthcare coverage, but that in others it leads to a decrease in overall healthcare coverage. 
	
	\subsubsection{Price of insurance exceeds treatment price without insurance.}
	We start off by giving an example where the Nash equilibrium price of the insurance is higher than what would be the equilibrium drug price in the setting where there is no insurance at all. In particular, it follows that the overall healthcare coverage, that is $\muA+\muT$, is lower with insurance than it would be without insurance. 
	
	Consider $\mu=\frac 12 (\delta_{(0, 1)} + \delta_{(1, 1.9)})$. Without an insurance, the producer sets the treatment price at $1$ in order to get the maximum profit of $1$. In the presence of an insurer, a subgame-perfect Nash equilibrium must have the price of treatment at $\frac{1.9}r$ and the premium at $1.9$. In particular, in this case the premium is larger than the price of treatment if no insurer enters the market.
	
	A problem with the presented example is that with the above $\mu$, no Nash equilibrium exists because of the lack of continuity of $\Ppharma$ and $\Pinsur$ (a similar case was already discussed in Section \ref{sect:no-Nash-equilibria}). However, this problem is readily fixed by replacing Dirac deltas at $(x, y)$ with a uniform distribution on a small open neighborhood of $(x,y)$. 
	
	\subsubsection{Price of treatment decreases through introduction of insurance.}
	As seen in Theorem~\ref{thm:existence-Nash-equilibrium-with-insurer}, the profit of the producer increases in equilibrium when there is an insurer rather than not. However, the corresponding drug prices can be higher or lower, depending on the measure $\mu$. In Section \ref{sect:numerical-results}, we give examples of measures where the treatment price increases in the numerically found Nash equilibria. Furthermore, in the previous example, the treatment price goes up by a factor of $\frac 1r$.

	We now give a sketch  of a measure $\mu$ where the treatment price decreases. In particular, the overall coverage increases by introducting an insurer in this case.
	Consider the measure $\mu=\frac 12\delta_{(0, 1/r)} + \frac 12\delta_{(1, 1)}$ with a disease incidence $r<\frac 12$. If the insurer does not enter the market, the local maxima of producer's profits $\Ppharma$ occur at a treatment price of $\theta = \frac 1r$ with a profit of $\frac 12$ (this is the global maximum, hence the producer's choice) and at $\theta=1$ with a profit of $r<\frac 1 2$. 
	However, if we add an insurer in the game, then the producer is incentivized to lower the price infinitesimally below $\frac 1r$ so that the insurer is able to cover also the customers located at $(1,1)$ with positive profit. 
	The same problem of lack of continuity of $\Ppharma$ and $\Pinsur$ as with the previous example occurs, with a fix that is analogous to the one described there.

	\section{Numerical approach}\label{sect:numerical-results}
	The analytic search for Nash equilibria (NE) of the model defined in Section~\ref{sect:math-model} is intractable except in simple cases, hence we propose some numerical illustrations.

	\subsection{Description of numerical methods}\label{sect:description-of-methods}
	
	We consider the case where distribution $\mu$ has a density that is twice continuously differentiable and satisfies the assumptions of Lemma \ref{cor:existence-Nash-equilibrium}. By Theorem~\ref{thm:existence-Nash-equilibrium} there exists a subgame-perfect equilibrium which maximizes the profit of the producer among all NEs. 
	
	For subgame-perfect NEs, we first find  local maxima using:
	\begin{equation*}
		\begin{cases}
			\text{maximize }   & \Ppharma(\theta, \tildepi) \\
			\text{subject to } & \frac{\partial \Pinsur}{\partial \tildepi}(\theta, \tildepi) = 0 , \\
			& \frac{\partial^2 \Pinsur}{\partial\tildepi^2}(\theta, \tildepi) \le 0 , \\
			& r \theta \leq \tildepi \leq \theta .
		\end{cases}
	\end{equation*}
	In a second step, we select the maximum of the found local maxima.

	We also search for those NEs that maximize the insurer's profit, among all NEs;   they can be realized as dictatorial equilibria.\footnote{Note that other dictatorial equilibria, which do not necessarily maximize the insurer's profit, might also exist but are not the focus of this numerical investigation.} 
	The optimization problem writes:
	\begin{equation*}
		\begin{cases}
			\text{maximize }   & \Pinsur(\theta, \tildepi) \\
			\text{subject to } & r \theta \leq \tildepi \leq \theta , \\
			& \Ppharma(\theta, \tildepi) \geq \sup_{\theta'} \Ppharma(\theta', +\infty) .
		\end{cases}
	\end{equation*}

	\subsection{Numerical results}\label{sect:numerical-results-results}
	
	In order to test our methods, we consider different families of distributions $\mu$ for the characteristics of the population; for simplicity we take them of the form $\mu=\mu_p\otimes\mu_\psi$, that is, the subjective probability of falling sick $p$ and the drug reservation price $\psi$ of a randomly selected agent are independent, and they are $\mu_p$- and $\mu_\psi$-distributed, respectively.
	
	\begin{center}\small 
		\begin{tabular}{l|l|l|l}
			Section & $\mu_p$ & $\mu_\psi$ & equilibrium type \\
			\hline 
			\ref{sect:varying-exponential} & $\mathrm{Beta}(2,3)$ & $\mathrm{Exp}($various$)$ & subgame-perfect \\
			\hline
			\ref{sect:varying-Beta-exponential} & $\mathrm{Beta}($various$,2)$ & $\mathrm{Exp}(1)$ & subgame-perfect \\
			\hline
			\ref{sect:dictatorial-eq} & $\mathrm{Beta}($various$,2)$ & $\mathrm{Exp}(1)$ & dictatorial \\
			\hline
			\ref{sect:varying-Pareto} & $\mathrm{Beta}(2,2)$ & Pareto$($various$)$ & subgame-perfect
		\end{tabular}
	\end{center}

	\subsubsection{Varying exponential families}\label{sect:varying-exponential}
	
	The subgame-perfect NEs $(\theta^*,\tildepi^*)$ for $\mu_p = \mathrm{Beta}(2,3)$ a Beta distribution with shape parameters $2$ and $3$ and $\mu_\psi=\mathrm{Exp}(\lambda)$ an exponential distribution with varying parameter $\lambda$ are plotted in Figure \ref{fig:Beta-varying-exponential-equilibria}, together with the proportion of patients choosing the various options and the profit by the insurer and producer as a function of $\lambda$. We observe that the profits and the premium and treatment price are inversely proportional to $\lambda$, while the fraction of people choosing each option remains constant; indeed, all these problems reduce to the same one (with $\mu_\psi = \mathrm{Exp}(1)$) by multiplying all costs by $\lambda$. Note that in this example,  $1/\lambda$ represents the average reservation price in the population.
	
	\begin{table}\begin{center}
	    \captionof{table}{Beta-varying-exponential-equilibria} \label{tab:Beta-varying-exponential-equilibria}
		\footnotesize
		\begin{tabular}{c|cc|ccc|cc}
			\toprule
			$\lambda$ & $\theta^*$ & $\tildepi^*$ & $\mathcal A^*$ & $\mathcal T^*$ & $\mathcal O^*$ & $\Pinsur(\theta^*,\tildepi^*)$ & $\Ppharma(\theta^*,\tildepi^*)$ \\
			\midrule
			1.000 & 1.113 & 0.540 & 0.232 & 0.157 & 0.611 & 0.048 & 0.130 \\
			1.571 & 0.708 & 0.344 & 0.232 & 0.157 & 0.611 & 0.030 & 0.083 \\
			2.143 & 0.519 & 0.252 & 0.232 & 0.157 & 0.611 & 0.022 & 0.061 \\
			2.714 & 0.410 & 0.199 & 0.232 & 0.157 & 0.611 & 0.018 & 0.048 \\
			3.286 & 0.339 & 0.164 & 0.232 & 0.157 & 0.611 & 0.015 & 0.040 \\
			3.857 & 0.288 & 0.140 & 0.232 & 0.157 & 0.611 & 0.012 & 0.034 \\
			4.429 & 0.251 & 0.122 & 0.232 & 0.157 & 0.611 & 0.011 & 0.029 \\
			5.000 & 0.223 & 0.108 & 0.232 & 0.157 & 0.611 & 0.010 & 0.026 \\
			\bottomrule
		\end{tabular}
	\end{center}\end{table}
	\begin{figure}[H]  
		\centering
		\includegraphics[width=0.7\linewidth]{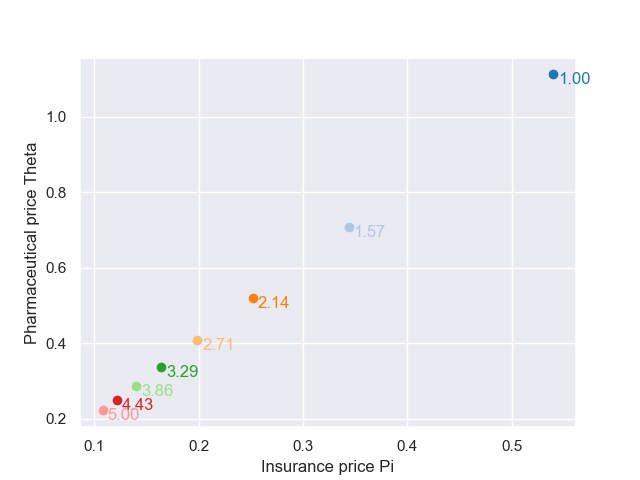}
		\newline
		\subfloat{
			\includegraphics[width=0.5\linewidth]{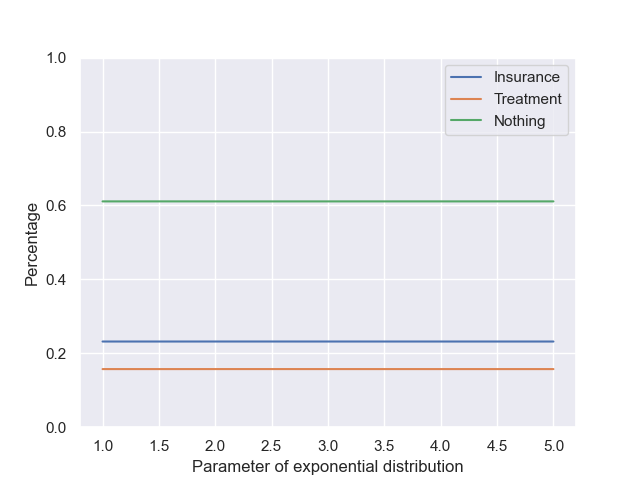}
		}
		\subfloat{
			\includegraphics[width=0.5\linewidth]{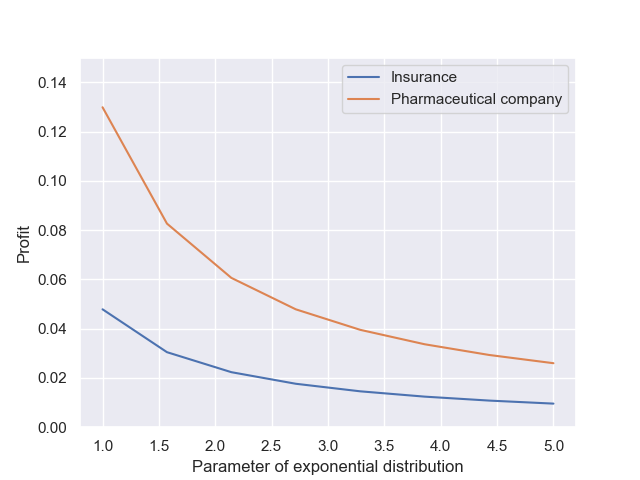}
		}
		\caption{Top:  NEs for $\mu=\mathrm{Beta}(2,3) \otimes\mathrm{Exp}(\lambda)$, with $\lambda$ 
			indicated as annotation next to each point in the plot. We use a disease incidence of $r=0.3$. 
			Bottom left: percentage of the population choosing insurance, treatment, or choosing neither, at the Nash equilibrium, as a function of $\lambda$. Bottom right plot: profit of the insurance and the producer at equilibrium as a function of $\lambda$ (constants for this choice of distribution). 
		}
		\label{fig:Beta-varying-exponential-equilibria}
	\end{figure}
	
	\subsubsection{Varying Beta families}\label{sect:varying-Beta-exponential}
	
	Figure \ref{fig:Varying-Beta-exponential-equilibria} shows subgame-perfect NEs for $\mu = \mathrm{Beta}(s_1, 2) \otimes \mathrm{Exp}(1)$ for various values of $s_1$. 
	We observe that as the first shape parameter of the Beta distribution tends to $0$, the equilibrium treatment price tends to $1$, which is the price that optimizes the producer profit if there is no insurance. This behavior is expected: when the first shape parameter tends to $0$, the Beta distribution concentrates around $0$. Since $\tildepi \geq r\theta$ no patient $(p,\psi)$ with $p<r$ buys insurance. Thus, when $s_1 \to 0$, a vanishing number of people buy insurance. 
	Furthermore, increasing $s_1$ (which makes higher $p$ more likely in the population) leads to an increase of the premium and of the treatment price. Heuristically, a large $s_1$ makes the Beta distribution concentrated around $1$; more people are willing to buy (and pay more for) insurance given that they deem it likely to fall sick, and the higher demand also allows the producer to increase its prices since most of the coverage goes through the insurance anyway.
	
	\begin{table}\begin{center}\captionof{table}{Varying-Beta-exponential-equilibria} \label{tab:Varying-Beta-exponential-equilibria}
		\footnotesize
		\begin{tabular}{c|cc|ccc|cc}
			\toprule
			$s_1$ & $\theta^*$ & $\tildepi^*$ & $\mathcal A^*$ & $\mathcal T^*$ & $\mathcal O^*$ & $\Pinsur(\theta^*,\tildepi^*)$ & $\Ppharma(\theta^*,\tildepi^*)$ \\
			\midrule
			0.100 & 1.008 & 0.485 & 0.010 & 0.357 & 0.633 & 0.002 & 0.111 \\
			0.200 & 1.015 & 0.491 & 0.021 & 0.346 & 0.633 & 0.004 & 0.112 \\
			0.500 & 1.041 & 0.514 & 0.054 & 0.311 & 0.636 & 0.011 & 0.114 \\
			1.000 & 1.088 & 0.553 & 0.105 & 0.256 & 0.640 & 0.024 & 0.118 \\
			1.500 & 1.133 & 0.593 & 0.147 & 0.209 & 0.644 & 0.037 & 0.121 \\
			2.000 & 1.173 & 0.630 & 0.180 & 0.172 & 0.648 & 0.050 & 0.124 \\
			3.000 & 1.239 & 0.698 & 0.226 & 0.119 & 0.654 & 0.074 & 0.129 \\
			4.000 & 1.286 & 0.754 & 0.254 & 0.087 & 0.659 & 0.093 & 0.131 \\
			5.000 & 1.316 & 0.799 & 0.271 & 0.066 & 0.663 & 0.110 & 0.133 \\
			10.000 & 1.344 & 0.925 & 0.305 & 0.025 & 0.669 & 0.159 & 0.133 \\
			\bottomrule
		\end{tabular}
        \end{center}
	\end{table}
	
	\subsubsection{Dictatorial Nash equilibria}\label{sect:dictatorial-eq}
	Figure \ref{fig:Beta-times-exponential-dictatorial} illustrates dictatorial NEs (see Definition \ref{def:dictatorial}). For $\mu=\mathrm{Beta}(2,2)\otimes\mathrm{Exp}(\lambda)$ and $\mu=\mathrm{Beta}(2,2)\otimes\text{Pareto}(1, s_2)$ we find the results are very similar as those found in sections \ref{sect:varying-Beta-exponential} and \ref{sect:varying-Pareto}, respectively. However, the behavior of dictatorial NEs is markedly different for $\mu=\mathrm{Beta}(s_1, 2)\times\mathrm{Exp}(1)$. We thus highlight only this case here.
	
	Coverage and treatment curves are similar in Figures~\ref{fig:Varying-Beta-exponential-equilibria} and~\ref{fig:Beta-times-exponential-dictatorial}. However, the price of treatment starts decreasing when $s_1$ is large in Figure~\ref{fig:Beta-times-exponential-dictatorial}, while it is always increasing in Figure~\ref{fig:Varying-Beta-exponential-equilibria}.
	\begin{table}\begin{center}\captionof{table}{Beta-times-exponential-dictatorial} \label{tab:Beta-times-exponential-dictatorial}
		\footnotesize
		\begin{tabular}{c|cc|ccc|cc}
			\toprule
			$s_1$ & $\theta^*$ & $\tildepi^*$ & $\mathcal A^*$ & $\mathcal T^*$ & $\mathcal O^*$ & $\Pinsur(\theta^*,\tildepi^*)$ & $\Ppharma(\theta^*,\tildepi^*)$ \\
			\midrule
			0.400 & 1.321 & 0.639 & 0.036 & 0.241 & 0.724 & 0.009 & 0.110 \\
			0.500 & 1.335 & 0.650 & 0.044 & 0.231 & 0.725 & 0.011 & 0.110 \\
			1.000 & 1.332 & 0.670 & 0.089 & 0.199 & 0.712 & 0.024 & 0.115 \\
			1.500 & 1.330 & 0.689 & 0.130 & 0.170 & 0.701 & 0.038 & 0.119 \\
			2.000 & 1.328 & 0.707 & 0.164 & 0.145 & 0.691 & 0.051 & 0.123 \\
			3.000 & 1.317 & 0.738 & 0.215 & 0.109 & 0.675 & 0.074 & 0.128 \\
			4.000 & 1.306 & 0.765 & 0.251 & 0.085 & 0.665 & 0.093 & 0.131 \\
			5.000 & 1.296 & 0.788 & 0.275 & 0.067 & 0.658 & 0.110 & 0.133 \\
			10.000 & 1.238 & 0.858 & 0.328 & 0.030 & 0.642 & 0.160 & 0.133 \\
			\bottomrule
		\end{tabular}
        \end{center}
	\end{table}

	\subsubsection{Varying Pareto families}\label{sect:varying-Pareto}
	Finally, we consider $\mu=\mathrm{Beta}(2,2) \otimes \mathrm{Pareto}(1,s_2)$ for various values of $s_2$ in Figure \ref{fig:Beta-varying-Pareto}. The higher $s_2$, the more the Pareto distribution is concentrated around $1$.

	\begin{table}\begin{center}\captionof{table}{Varying Pareto families} \label{tab:Varying Pareto families}
		\footnotesize
		\begin{tabular}{c|cc|ccc|cc}
			\toprule
			$s_2$ & $\theta^*$ & $\tildepi^*$ & $\mathcal A^*$ & $\mathcal T^*$ & $\mathcal O^*$ & $\Pinsur(\theta^*,\tildepi^*)$ & $\Ppharma(\theta^*,\tildepi^*)$ \\
			\midrule
			1.100 & 1.822 & 0.964 & 0.335 & 0.293 & 0.373 & 0.140 & 0.343 \\
			1.200 & 1.683 & 0.896 & 0.350 & 0.306 & 0.343 & 0.137 & 0.332 \\
			1.400 & 1.392 & 0.764 & 0.378 & 0.378 & 0.244 & 0.131 & 0.316 \\
			1.600 & 1.158 & 0.654 & 0.393 & 0.493 & 0.114 & 0.121 & 0.308 \\
			1.800 & 1.068 & 0.608 & 0.396 & 0.556 & 0.048 & 0.114 & 0.305 \\
			2.000 & 1.041 & 0.592 & 0.396 & 0.577 & 0.026 & 0.111 & 0.304 \\
			3.000 & 1.015 & 0.578 & 0.396 & 0.596 & 0.007 & 0.108 & 0.302 \\
			4.000 & 1.010 & 0.575 & 0.396 & 0.599 & 0.004 & 0.108 & 0.302 \\
			5.000 & 1.008 & 0.574 & 0.396 & 0.601 & 0.003 & 0.108 & 0.302 \\
			6.000 & 1.007 & 0.573 & 0.396 & 0.601 & 0.002 & 0.108 & 0.301 \\
			\bottomrule
		\end{tabular}
        \end{center}
	\end{table}

	\begin{figure}[H]  
		\centering
		\includegraphics[width=0.7\linewidth]{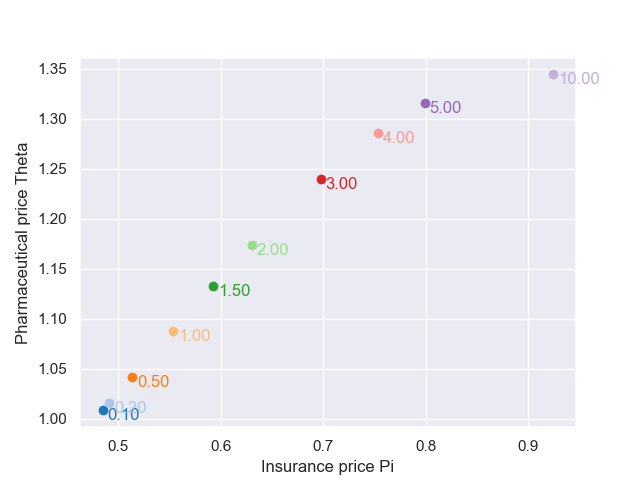}
		\newline
		\subfloat{
			\includegraphics[width=0.5\linewidth]{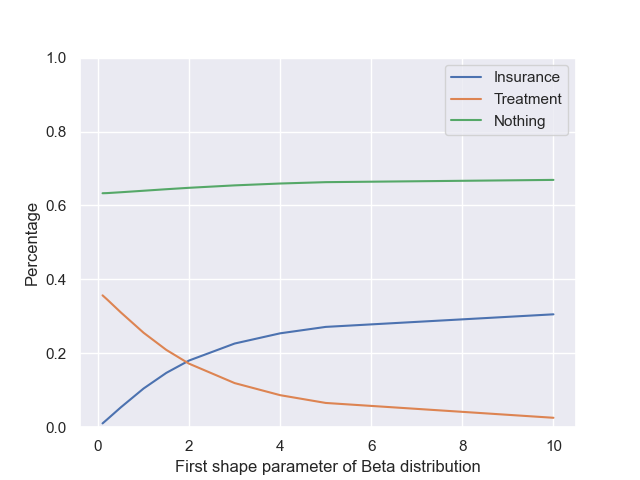}
		}
		\subfloat{
			\includegraphics[width=0.5\linewidth]{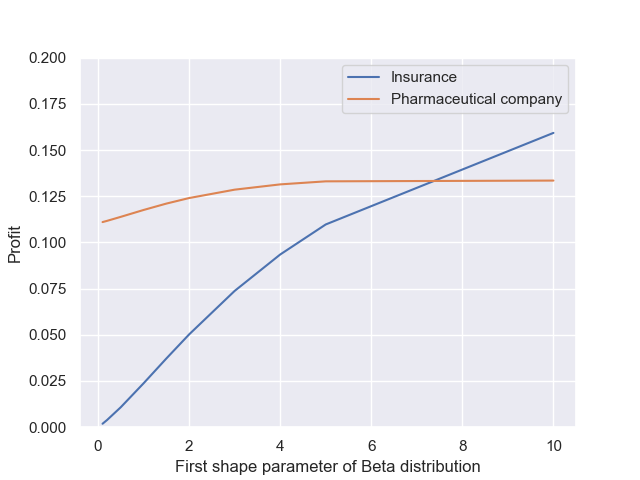}
		}
		\caption{This figure is analogous to Figure \ref{fig:Beta-varying-exponential-equilibria} but with $\mu=\mathrm{Beta}(s_1, 2)\otimes\mathrm{Exp}(1)$, where the parameter $s_1$ is indicated as an annotation next to each point. The disease incidence is $r=0.3$.} 
	\label{fig:Varying-Beta-exponential-equilibria}
\end{figure}

\begin{figure}[H]  
	\centering
	\includegraphics[width=0.7\linewidth]{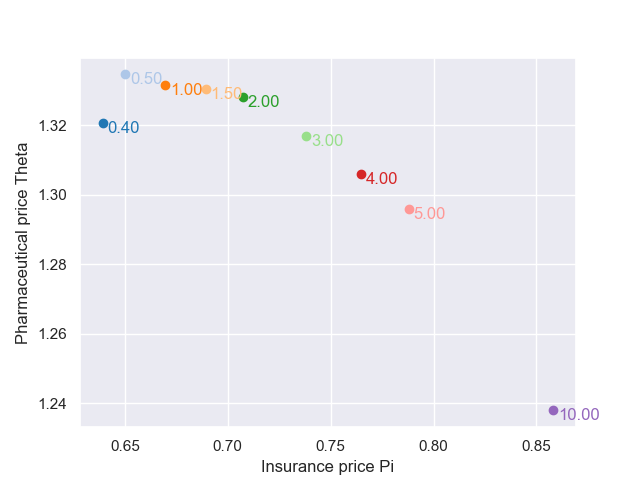}
	\subfloat{
		\includegraphics[width=0.5\linewidth]{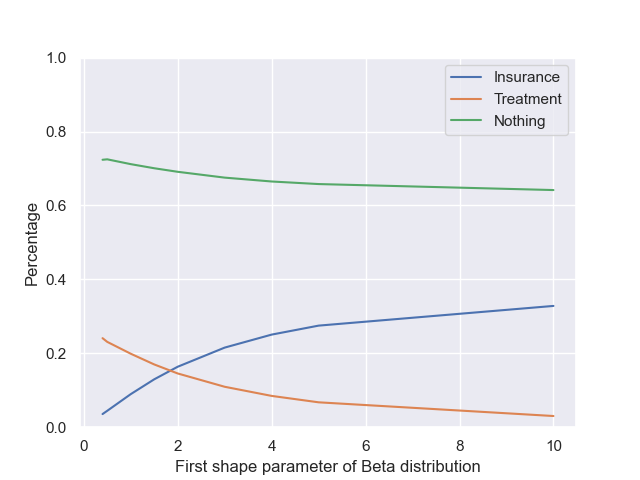}
	}
	\subfloat{
		\includegraphics[width=0.5\linewidth]{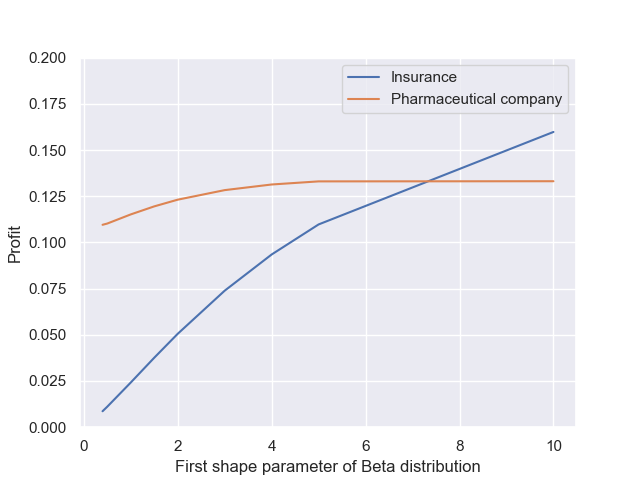}
	}
	\caption{This figure is analogous to Figure \ref{fig:Varying-Beta-exponential-equilibria}, with the difference that it shows dictatorial Nash equilibria (instead of subgame-perfect equilibria). 
		We consider $\mu=\mathrm{Beta}(s_1,2)\times\mathrm{Exp}(1)$, where 
		the shape parameter $s_1$ is indicated in the plot next to each point. The incidence rate is $r=0.3$.} 
\label{fig:Beta-times-exponential-dictatorial}
\end{figure}

\begin{figure}  
\centering
\includegraphics[width=0.7\linewidth]{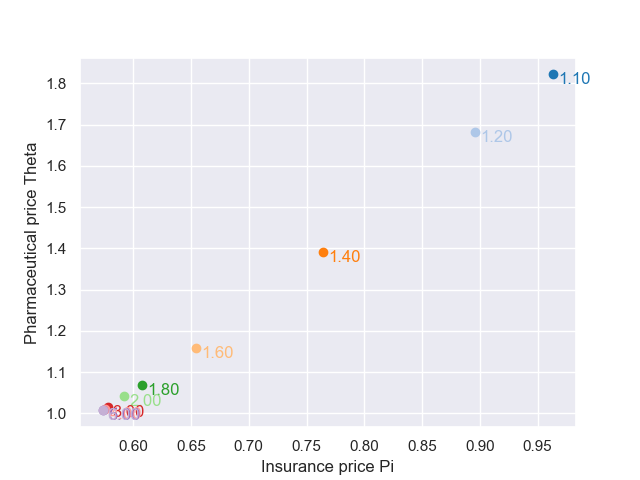}
\newline
\subfloat{
	\includegraphics[width=0.5\linewidth]{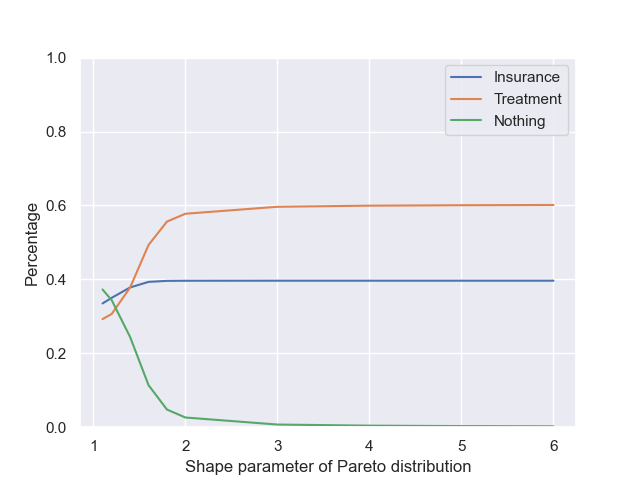}
}
\subfloat{
	\includegraphics[width=0.5\linewidth]{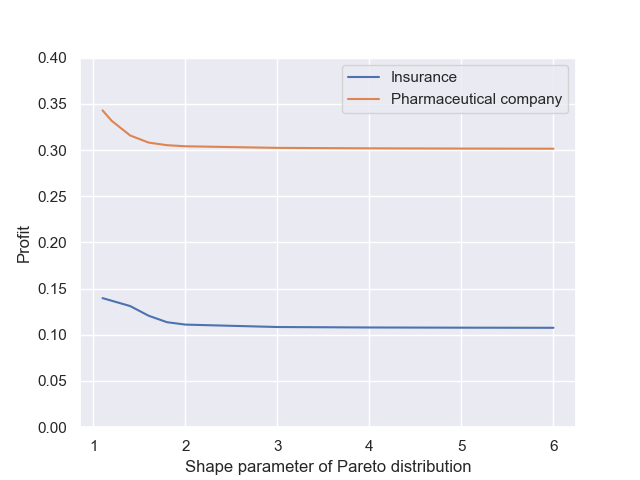}
}
\caption{This figure is analogous to Figure \ref{fig:Beta-varying-exponential-equilibria} but with $\mu=\mathrm{Beta}(2, 2)\otimes\text{Pareto}(1, s_2)$, where the scale parameter of the Pareto distribution (that is, the infimum of the support of the distribution) is $1$, and $s_2$ denotes the shape parameter of the Pareto distribution. The density of the Pareto distribution is thus $t\mapsto \mathds 1_{t\ge 1} \frac{s_2}{t^{s_2+1}}$. The parameter $s_2$ is indicated as an annotation next to each point. The disease incidence is $r=0.3$.}
\label{fig:Beta-varying-Pareto}
\end{figure}

\section{Conclusions and Perspectives}
This paper integrates monopoly power, insurer strategy, and patient heterogeneity into a single sequential‐game framework for life‑saving drug markets. By allowing both the manufacturer and the insurer to choose prices endogenously, we uncover that insurance guarantees greater producer surplus, yet can raise or lower treatment prices and access, depending on the joint distribution of illness risk and liquidity constraints. Our analytical results, reinforced by calibrated simulations, highlight policy levers—premium regulation, price caps, and value‑based subsidies -- that can realign incentives without bluntly suppressing innovation. More broadly, the model provides a tractable benchmark for evaluating ongoing proposals such as Medicare price negotiation, outcome‑based contracts, and international reference pricing. Future empirical work that recovers the sufficient statistics we identify will be critical for translating these theoretical insights into actionable, welfare‑improving reforms.

\bibliographystyle{abbrv}
\bibliography{bib.bib}

\appendix
\section{Optimal health decisions and the value of life}\label{sec:multiperiod}

In this appendix, we propose a multi-period health decision problem for each agent in the population, extending the single-period model in Subsection \ref{sect:Customers} (in particular Example \ref{example1}) and incorporating  long-term implications of medical treatment choices. Lifetime utility allows for uncertainty in future income, consumption, discounting, and risk aversion, yet the model remains tractable: in the end, the population is summarised, as before, by two sufficient statistics per period.  The producer–insurer game therefore extends naturally to multi‑period horizons without loss of analytical clarity.



At each period $t\in\mathbb N$, a population of $S_t$ \emph{susceptible} agents may (i) buy full insurance for premium $\tilde\pi_t$, (ii) face diagnosis, and (iii) if not covered, decide on treatment at price $\theta_t$.  Once diagnosed and treated (or not), the agent exits the pool. 



Let us describe the decision problem of a single agent, called agent $a$, susceptible at time 0.  We introduce a measurable space $(\Omega_F, \widetilde{\mathcal F})$ to model financial randomness and a second one $(\Omega_H, \mathcal H)$, for the health outcomes of the agent. Economic and health-related random variables co-exist in the product space $(\Omega,\mathcal F):=(\Omega_F\times \Omega_H,\widetilde{\mathcal F}\otimes\mathcal H )$, endowed with a probability $\mathbb P$. The model is based on the following agent-specific elements:
\begin{itemize}
\item[-] the agent's time of death, a $\mathcal H$-measurable random variable, $T\in \mathbb R_+^*$;

\item[-] the agent's diagnosis time, a $\mathcal H$-measurable random variable denoted by $\tau$, $\mathbb R^*_+\cup \{+\infty\}$-valued and having a continuous distribution;


\item[-] agent-level information. We assume the agent observes both economic and health-related factors. Thus, at time $t\in\mathbb N$, the agent observes the prices of drugs and insurance prevailing for the period $[t,t+1)$, as well as the predictable consumption of the period, denoted by $c_t$. All these are assumed $\widetilde{\mathcal F}$-measurable and for now also assumed exogenous to the analysis. Denote:
$$
\widetilde{\mathcal F_t}:=\sigma(\theta_s,\pi_s, c_s, s\in \{0,..., t\}).
$$
The diagnosis time $\tau$ is also revealed, once it occurs. For simplicity, we assume this the only relevant health-related information. Hence, denote the agent's information at time $t\in \mathbb N$ by 
\begin{equation}\label{eq:filtration}
	\mathcal F_t:=\widetilde{\mathcal F_t} \otimes \sigma(\tau\wedge s, s\in\{0,..., t\}).
\end{equation}

Thus, $\mathcal F_t$ represents the information of the agent at the beginning of the period $[t,t+1)$, that is relevant for decisions to be taken at time $t$. The intra-period information of the agent is given by
$$
\mathcal F_{t_+}:= \widetilde{\mathcal F_t} \otimes \sigma(\tau \wedge (t+1)),
$$
hence $\mathcal F_t\subset \mathcal F_{t_+}\subset \mathcal F_{t+1}$. The interpretation is as follows. The diagnosis time $\tau$ can a priori occur within any period. This may trigger unexpected consumption, for example, if the agent decides to buy the treatment out of pocket at the diagnosis time $\tau$. We need intra-period information in order to model such an intra-period decision. 
\end{itemize}

\subsection{Preferences}\label{sec:pref} 
Each agent in the population is characterized by preferences. We assume that an agent's future health prospects affect preferences for consumption choices in upcoming periods. This assumption is inspired by \cite{Pliskin} and \cite{murphy2006}.
\begin{definition}\label{def:cons}
A \textit{consumption stream} of agent $a$ is a process $(\tilde c_t)_{t\in\mathbb N}$, adapted with respect to $(\mathcal F_{t+})_{t\in\mathbb N}$,   where $\tilde c_t$ represents the consumption of the period $[t,t+1)$ 
$$
\tilde c_t=c_t +\Delta c_t ,
$$
where $(c_t)_{t\in\mathbb N},$ is the predictable consumption ---exogenous to the analysis--- and $\Delta c_t\leq 0$ 
is interpreted as a health-related spending of the period $[t,t+1)$. 
\end{definition}
Note that the health expenses $\Delta c_t$ are intended to result from a decision process, as detailed in Subsection \ref{sec:decisions}, and we will consider consumption streams with $\Delta c\in\{0,\pi_t,\theta_t\}$, that is, no expense, purchase of insurance, or purchase of treatment.

To assess the agent's consumption during a future period, relatively to the consumption of a previous period,
we introduce a family of utility functions $u^H_{s,t}: L^\infty(\Omega, \mathcal F_{t_+})\to L^\infty(\Omega, \mathcal F_{s})$, $s\leq t\in \mathbb N$, assumed to encode the agent's preferences, and represented as:
\begin{equation}\label{eq:utility1}
u^H_{s,t}(X)= \frac{\mathbf 1_{T>s}}{(1+\alpha)^{t-s} }\mathbb E_{\mathbb Q_s}\big  [X \prod_{\ell=s}^{t}q_{\ell}\gamma_{\ell} \big |\mathcal F_s\big ]
\end{equation} where:
\begin{itemize} 
\item[-] the constant $\alpha>0$ is the discount factor in the economy; 
\item[-] the process $H=(H_t)_{t\in\mathbb N}$ appearing as an index for the utility functions indicates the \textit{health status} of the agent $a$ at any time $t$, assumed to have the representation: 
$$
H_t:= \prod_{\ell =1}^t q_{\ell}\gamma_\ell, \quad t\in \mathbb N,
$$ where
\begin{itemize}
	\item[-] $q_t\in(0,1)$ is an $(\mathcal F_t)$-adapted process, representing the loss/improvement in the quality of life during the time interval $[t,t+1)$. Hence $\prod_{\ell=1}^kq_{t+\ell}$ represents the relative loss/improvement in the quality of life $k$ periods after time $t$; 
	\item[-] $\gamma_t\in(0,1)$ is is an $(\mathcal F_t)$-adapted process, i.e., the probability of survival during the period $[t,t+1)$, given survival at $t$. Hence, $\prod_{\ell=s}^t\gamma_{\ell}$ represents the probability of survival during the period $[t,t+1)$, given survival at time $s$; 
	
\end{itemize}
\item[-] $(\mathbb Q_t)_{t\in\mathbb N}$ is a family of  probability measures absolutely continuous with respect to $\mathbb P$, that are specific to  agent $a$. Intuitively,  each  $\mathbb Q_t$ represents  time-dependent ``scenarios weights", that may reflect the level of risk aversion of the agent at time $t$. 

\end{itemize}
\begin{remark}
The agent's utility at time $t$  is based upon  expectation under some subjective probability, that can be updated as time evolves, not only to account for additional information (through the conditionning) but also for changes in beliefs. As an example,  assume that agent  becomes more pessimistic with aging,  regarding the evolution of health and consumption. This feature can be captured by probabilities $\mathbb Q_t$ that induce first order stochastic dominance for health and consumption variables, as follows: $\mathbb Q_s(c_t<x|\mathcal F_u)\leq \mathbb Q_u(c_t<x|\mathcal F_u)$ and $\mathbb Q_s(q_t\gamma_t<x|\mathcal F_u)\leq \mathbb Q_u(q_t\gamma_t<x|\mathcal F_u)$, for $s\leq u\leq t$. Similarly,  we could imagine individuals that become more pessimistic about health prospects but more optimistic about consumption, etc.
\end{remark}
\begin{remark}[Inter-period utilities]The random variable $X$ in \eqref{eq:utility1} is $\mathcal F_{t_+}$-measurable, with the interpretation is that $\tilde c_{t}:= X$ corresponds to a consumption during the interval $[t,t+1)$. The conditioning wrt $\mathcal F_s$ indicates that \eqref{eq:utility1} represents an assessment at time $s$, or, more exactly, at the beginning of the period $[s,s+1)$. We define an intra-period assessment  in a similar manner by conditioning wrt $\mathcal F_{s_+}$:
\begin{equation}\label{eq:utilitybis}
	u^H_{s_+,t}(X)= \frac{\mathbf 1_{T>s}}{(1+\alpha)^{t-s} }\mathbb E_{\mathbb Q_s} \big [X    \prod_{\ell=s}^tq_{\ell}\gamma_{\ell} \big |\mathcal F_{s_+} \big ]
\end{equation}where we assume that the agent does not update the probabilities within a period.
\end{remark}
\begin{remark}\begin{itemize}
	\item[i.] Our utility functions are {\it monetary} because they satisfy: for $m\in \mathcal F_t$,  $u^H_{t,t}(m)= m$ for $t<T$ (see Delbaen \cite{Delb12} for more on such utility functions). 
	We adopt here a simplified numerical representation of preferences, convenient for computing the drug reservation price (cf. Subsection \ref{sec:multipsi}).

	\item[ii.] In our approach, the health status evolution acts as a stochastic discount factor. Indeed, the quantity $q_t\gamma_t/(1+\alpha)\in (0,1)$ `` discounts'' the consumption of period $[t,t+1)$, in order to be compared to the consumption of period $[t,t+1)$. 
\end{itemize}
\end{remark}

\begin{definition}Given a consumption flow  $\tilde c = (\tilde c_t)_{t\in\mathbb N}$, and  health status evolution $H= (H_t)_{t\in\mathbb N}$, the \textit{global utility} of the agent at any time $t\in\mathbb N$ is defined as as:
$$
U_t(\tilde c,H)=\sum_{k=0}^\infty u_{t,t+k}^H(\tilde c_k);
$$ and the global utility at the agent's diagnosis time is defined as:
$$
U_\tau(\tilde c,H):=\sum_{t=0}^\infty\sum_{k=0}^\infty \mathbf 1_{\tau\in(t,t+1]}u_{t_+,t+k}^H(\tilde c_k).
$$
\end{definition}
The interpretation is that $U_t$ measures the utility at time $t$ of the consumption in the remaining lifetime of the agent, $[t,T)$, while the utility at the diagnosis time $U_\tau$ measures the cumulative utility from the diagnosis time on.

\subsection{Health-impacting decisions}\label{sec:decisions}

We now formalise the idea that the agent may wish to some extent to ``exchange'' units of consumption for units of health. This can be achieved via some actions, such as buying insurance before diagnosis, or, alternatively, buying the treatment upon diagnosis.

A decision process $\delta= (\delta_{t})_{t\in\mathbb N}$ is a stochastic process $\{0,1\}$-valued and $(\mathcal F_t)_{t\in\mathbb N}$ adapted. Decisions impact the consumption process $(c^\delta_t)_{t\in\mathbb N}$ of the agent and their health status $(H^\delta_t)_{t\in\mathbb N}$, as detailed below.

For a decision process $\delta$, we use the convention $\delta_t=0$ for ``no action is taken''; $\delta_t=1$ for ``action is taken''. 
The interpretation is as follows:
\begin{itemize}
\item[-] If $\tau>t $ (i.e., before diagnosis), we set $\delta_t=1$ if the agent decides to buy insurance covering treatment costs for $\tau\in (t,t+1]$. Alternatively, $\delta_t=0$ means that the agent has no insurance covering the period $(t,t+1]$. 
\item[-] If $\tau\in (t-1,t]$ (i.e., upon diagnosis) we set $\delta_t=1$ if the agent decides to buy the treatment out of pocket. In this case, the consumption of the period $[t-1,t)$ is impacted. 

Alternatively, we set $\delta_t=0$, if the agent does not buy the treatment during the period $(t-1,t]$, despite a diagnosis. 

\item[-] Tf $\tau<t-1$ (i.e., in the periods starting after the diagnosis time)  no more decisions are considered. We set $\delta_t=0$ (this condition will be incorporated in the admissibility conditions, see below).
\end{itemize}
We denote by $ \mathcal D$ the set of admissible decisions, defined as all $(\mathcal F_t)$-adapted, $\{0,1\}$-valued processes $(\delta_t)$,  with $\delta_{t+1}\mathbf 1_{\tau\in (t,t+1]}$ being $\mathcal F_{t+}$-measurable (in other words, the decision of whether to buy or not the treatment is an inter-period decision, interpreted as being taken at the diagnosis time) and also satisfying $|\{\delta_t=1, \tau\leq t\}|=0$. This means that the agent buys the treatment out of pocket maximum one time, no later than in the period spanning the diagnosis time. If the treatment is purchased, the treatment is fully followed, possibly for several periods, with possible recovery or not, but no action is possibly taken afterward.\footnote{In practice, some patients may take have more than one year to decide on whether to start a certain treatment. Here, we assume that such agents are of negligible mass in the overall population, so that they do not influence the pricing game.}

We denote by, or rather interpret, $(c_t)_{t\in\mathbb N}$ (resp. $(H_t)_{t\in\mathbb N}$) as standing for the consumption process (resp. the health status) corresponding to the decision $\delta \equiv 0$ (no action, that is, no insurance, nor treatment). We recall that $c_t$ is also called the predictable consumption of the agent during the period $[t,t+1)$. For a decision process $\delta\in\mathcal D$, we pose the following representation of its corresponding consumption stream:
$$
c^{\delta}_t = c_t-\delta_t\tildepi_t -\delta_{t+1}\theta_t\mathbf 1_{\tau\in(t,t+1]}
$$
that is, whenever a susceptible agent takes an action to buy insurance at time $t$, the consumption of the period $[t,t+1)$ is reduced by the cost of insurance $\tildepi_t$, and by the cost of treatment $\theta_t$, just after diagnosis if an action is taken at the diagnosis time. Of course, it is not optimal for the agent to buy insurance and also pay the treatment out of pocket in the insurance window, so we end up having $\delta_t\delta_{t+1}\mathbf 1_{\tau\in[t,t+1)}=0$ for optimal decisions.

The health status of the agent is also influenced by the decision process, but only after diagnosis. We assume that the agent views the treatment as having a positive impact on the health status. Denote $(H^{(\tau,1)}_t)$ the health status of the agent when starting the treatment at the diagnosis time $\tau$; and assume it evolves as
$$
H^{(\tau,1)}_t=
\begin{cases}
H_t(1+\epsilon_1(1+\epsilon_2)^{\lceil \tau \rceil -t }),&\text{ if } t>\tau\text{ and either }\delta_{\lfloor \tau \rfloor} =1, \text{ or } \delta_{\lceil \tau \rceil}=1\\
H_t,&\text{ otherwise }.\\
\end{cases}
$$
where $\epsilon_i$ are nonnegative constants, characterizing the drug efficacy, as perceived by the agent. The constant $\epsilon_1$ quantifies the impact on the quality of life, while $\epsilon_2$ is the improvement of the survival chances in one period (other specifications are possible). Recall that $H_t$ is the health status without treatment. Also note the agent can start the treatment if the decision to buy insurance was made just before the diagnosis (i.e., at time $t \in \mathbb N$ satisfying $\tau \in (t,t+1]$), or by a post-diagnosis action to buy the treatment. 

Our aim is to understand optimal actions  at a point in time, by studying the  problem:
$$
V_t:= \text{ess} \sup_{\delta\in \mathcal D} U_t(c^\delta,H^\delta), \quad t\in\mathbb N,
$$
where the essential supremum is with respect to the subjective probability $ \mathbb Q_t$.

\subsection{Subjective probabilities for diagnosis}

\begin{lemma}\label{lem:utility} For any $t,k\in\mathbb N$, $k>1$ and $\delta\in \mathcal D$ that involves taking a treatment if $\tau\in(t,t+1]$, we have:
$$\mathbf 1_{\{\tau>t\}} u^{H^{\delta}}_{t,t+k}(c^\delta_{t+k})=\mathbf 1_{\{\tau>t\}} ( p_t u^{H^{(\tau,1)}}_{t,t+k}(c_{t+k})+ u^{H^\delta}_{t,t+k}(\mathbf 1_{\{\tau>t+1\}} c^\delta_{t+k})),
$$ and 
$$
u^{H^\delta}_{t_+,t+k}(\mathbf 1_{\{\tau\in(t,t+1]\}}c^\delta_{t+k})=\mathbf 1_{\{\tau\in(t,t+1] \}}u^{H^{(\tau,1)}}_{t,t+k}(c_{t+k}),
$$ 
where 
$
p_t: =\mathbb Q_t(\tau\in(t,t+1]|\tau>t).
$
Similarily, if the decision $\delta$ involves not taking the treatment if $\tau\in(t,t+1]$ same expressions hold, but with $H^{(\tau,1)}$ being replaced with $H$ on the right hand sides of the equalities.
\end{lemma} 
\begin{proof}

By definition of the filtration $(\mathcal F_t)$, (see \eqref{eq:filtration}), any $\mathcal F_\infty$-measurable, random variable $X$ satisfies $\mathbf{1}_{\{\tau\in(t,t+1]\}}X=\mathbf{1}_{\{\tau\in(t,t+1]\}}\widetilde X_t$, for some $\widetilde X_t$ that is $\widetilde{\mathcal F}_\infty$-measurable. Using this representation and assuming $X$ is bounded, we can write:
\begin{align*}
	& \mathbf 1_{\{\tau>t\}} \mathbb E_{\mathbb Q} [X |\mathcal F_{t} ] = \mathbb E_{\mathbb Q}[\mathbf 1_{\{\tau\in(t,t+1]\}} \widetilde X_t |\mathcal F_{t} ] + \mathbb E_{\mathbb Q} [\mathbf 1_{\{\tau>t+1\}} X |\mathcal F_{t} ]\\
	&= \mathbf 1_{\{\tau>t\}}(\mathbb Q(\tau\in(t,t+1]|\tau>t) \mathbb E_{\mathbb Q}[ \widetilde X |\widetilde{\mathcal F}_{t}] + \mathbb E_{\mathbb Q} [\mathbf 1_{\{\tau>t+1\}} X |\mathcal F_{t}])\\
	&= \mathbf 1_{\{\tau>t\}}(p_t\mathbb E_{\mathbb Q} [ \widetilde X |\mathcal F_{t} ] + \mathbb E_{\mathbb Q} [\mathbf 1_{\{\tau>t+1\}} X |\mathcal F_{t} ]).
\end{align*} In the last step we have used the property: for any $\widetilde X\in\widetilde{\mathcal F}_\infty$, $\mathbb E_{\mathbb Q}[ \widetilde X |\widetilde{\mathcal F}_{t} ] =\mathbb E_{\mathbb Q}[ \widetilde X |\mathcal F_{t} ] $ that also follows directly from the representation of the filtration $(\mathcal F_t)$ in \eqref{eq:filtration}).

We now apply these results to the consumption process. For all $\delta\in \mathcal D$: $\tau \in(t,t+1]\Rightarrow c^\delta_{t+k}=c_{t+k}, k\geq 1$ (by assumption, there are no health expenses occurring more than one period after a diagnosis, as treatment can be bought at most once, upon diagnosis). Hence, 
$$
\mathbf 1_{\{\tau>t\}} c^\delta_{t+k}= \mathbf 1_{\{\tau\in(t,t+1]\}}c_{t+k}+\mathbf 1_{\{\tau>t+1\}}c^\delta_{t+k}
$$
Also, for decisions $\delta\in \mathcal D$ that involve taking a treatment if $\tau\in(t,t+1]$, 
$$
\mathbf 1_{\{\tau>t\}} H^{\delta}_{t+k}= \mathbf 1_{\{\tau\in(t,t+1]\}} H_{t+k}(1+\epsilon_1(1+\epsilon_2)^{-k })+\mathbf 1_{\{\tau>t+1\}} H^{\delta}_t,\quad \forall k
$$
The result follows from the definition of the utility functions, i.e., replacing $\mathbb Q$ with the agent's time $t$ subjective probability, $\mathbb Q_t$. 
\end{proof}

\subsection{Utility without insurance} \label{sec:multipsi}
Let us first assume there is no insurance (or, alternatively, that the agent does not buy an insurance) and focus on the decision at the diagnosis time. Assume that $\tau\in(t,t+1]<T$.

In case the agent does not buy the treatment, $\delta_{t+1}=0$, the health status evolves according to the exogenously given process $H$.
Hence, if the agent decides to go without treatment; the global utility at the diagnosis time writes:
\begin{align*} 
\mathbf 1_{\tau\in(t,t+1]} U_{\tau }(c,H) &= \mathbf 1_{\tau\in(t,t+1]} \left\{\sum_{k=1}^\infty u^H_{t_+,t+k}(c_{t+k})+c_t\right\}\\
&=\mathbf 1_{\tau\in(t,t+1]} \left\{ \sum_{k=1}^\infty u^H_{t,t+k}(c_{t+k})+c_t\right\}.
\end{align*}
(we have used Lemma \ref{lem:utility} for the second equality). Alternatively, the agent can decide to start treatment after diagnosis: $\delta_{t+1}=1$. In this case, the global utility at the diagnosis time writes:
\begin{align*}
\mathbf 1_{\tau\in(t,t+1]}U_{\tau} (c-\theta \mathbf 1_t,H^{(\tau,1)}) 
& = \mathbf 1_{\tau\in(t,t+1]} \left\{\sum_{k=1}^\infty u^{H^{(\tau,1)}}_{t_+,t+k}(c_{t+k})+c_t-\theta_t \right\}\\
& =\mathbf 1_{\tau\in(t,t+1]}\left\{ (1+\epsilon_1) \sum_{k=1}^\infty u^H_{t,t+k}(c_{t+k})(1+\epsilon_2)^{-k} +c_t- \theta_t\right\}.
\end{align*}

\begin{definition}
The reservation price of the treatment upon diagnosis is defined by: 
\begin{align*} 
	\psi 
	&= \sum_{t=1}^\infty \psi(t)\mathbf 1_{\tau \in (t,t+1]},
\end{align*}where we denoted
$
\psi(t):=\sum_{k=1}^\infty u^H_{t,t+k}(c_{t+k})\epsilon_1(1+\epsilon_2)^{-k}.
$
\end{definition}Notice that the drug reservation price corresponds to the additional utility from treatment upon diagnosis: $$\psi =U_{\tau}(c,H^{(\tau,1)})- U_{\tau}(c,H).$$ It can be checked that $\psi$ is bounded by $ \frac{\epsilon_1}{\epsilon_2}||c||_{\infty}$.
The agent, assumed to act as a utility maximizer, decides to buy the treatment if in the period where the agents is diagnosed, the reservation price exceeds the drug price that is, if $\psi\geq \theta_{\lfloor \tau\rfloor}$. We denote the health evolution under optimal decision upon diagnosis by $(H^{(\tau,*)}_t)$, i.e.
\begin{align*} 
H^{(\tau,*)}_t=
\begin{cases}
	H_t,& \quad \text{for } \tau \geq t \text{ or } \tau < t, \psi<\theta_{\lfloor \tau\rfloor}\\
	H^{(\tau,1)}_t,& \quad \text{for } \tau < t, \psi\geq \theta_{\lfloor \tau\rfloor}.
\end{cases}
\end{align*}
and by $c^{(\tau, *)}$, the consumption flow under optimal decision in absence of insurance:
\begin{align}\label{eq:cstar}
c^{(\tau,*)}_t= c_t -\theta_t\mathbf 1_{\{\psi(t)\geq \theta_t\}\cap\{\tau \in( t,t+1]\}}.
\end{align}

\begin{remark}
We notice that $\{\psi(t)\geq \theta_t\}\in\mathcal F_t\subset \mathcal F_{t_+}$, and $\{\tau \in(t,t+1]\}\in \mathcal F_{t_+}$. Hence, the consumption $(c^{(\tau,*)}_t)$ is in general $(\mathcal F_{t_+})$-adapted, as required by Definition \ref{def:cons}.
\end{remark}

\subsection{Utility with insurance options for each period}
We now incorporate insurance coverage options, at each time $t<\tau\wedge T$. A coverage decision can only take place at the beginning of an insurance period, that is, at time $t\in\mathbb N$, based on the information $\mathcal F_t$.

If buying insurance, $\delta_t=1$, the consumption of the period $[t,t+1)$ is $c_t-\tildepi_t$. If moreover the insured agent receives a diagnosis by time $t+1$, the health status evolves according to $H^{(\tau,1)}$. Note that past  coverage decisions impact the consumption of the relative period, but they do not affect the health status before $\tau$. So, at time $t$, unless $\tau\in(t,t+1]$, an insurance coverage decision has no impact on subsequent periods. This fact simplifies greatly the analysis of the optimal  decision problem.

Let $\tilde \delta_{t+1}, \tilde \delta_{t+2},...$, be some arbitrary decision process after time $t$, assumed to be followed by the agent, if no diagnosis occurs in the current period. Let $\tilde H$ and $\tilde c$ be the corresponding health and consumption of the agent under the decision process $\tilde\delta$. Assuming $\tilde \delta_t=1$ (insurance is bought at time $t$), and by using Lemma \ref{lem:utility}, we get the global utility at time $t$:
\begin{align*} 
\tilde U^1_{t}:&= c_t-\tildepi_t+ \sum_{k=1}^\infty  (u^{H^{(\tau,1)}}_{t,t+k} (c_{t+k})p_t+u^{\tilde H}_{t,t+k} (\tilde c_{t+k}\mathbf 1_{\tau>t+1} ))\\
&= c_t-\tildepi_t+ \sum_{k=1}^\infty u^{H^{(\tau,1)}}_{t,t+k} (c_{t+k})p_t+U_t (\tilde c\mathbf 1_{\tau>t+1},\tilde H ).
\end{align*}
Alternatively, assume the agent does not buy insurance at time $t$, that is $\tilde \delta_t=0$. Also assume  that if $\tau\in (t,t+1]$ the agent implements the optimal strategy after time $\tau$, and otherwise, if $\tau>t$, the agent implements the decision $\tilde \delta$. This leads to the global utility:
\begin{align*} 
\tilde U^0_{t}:=& 
\begin{cases}
	u^H_{t,t}\left (c_t -\theta\mathbf 1_{\{\tau \in(t,t+1]}\right) +\sum_{k=1}^\infty u^{H^{(\tau,1)}}_{t,t+k} (c_{t+k})p_t +U_t (\tilde c\mathbf 1_{\tau>t+1},\tilde H ),& \text{ if } \theta_t \leq \psi(t)\\
	u^H_{t,t}\left (c_t\right) +\sum_{k=1}^\infty u^{H^{(\tau,0)}}_{t,t+k} (c_{t+k})p_t +U_t (\tilde c\mathbf 1_{\tau>t+1},\tilde H ),& \text{ if } \theta_t > \psi(t)
\end{cases}
\end{align*}

Note $u^H_{t,t}\left (c_t -\theta\mathbf 1_{\tau \in(t,t+1]}\right) = c_t-\theta_t p_t$. Under the assumptions above, the agent buys insurance at time $t$ if and only if $\tau>t$ and 
$\tilde U^1_{t}\geq \tilde U^0_{t}$,
that is:
\begin{align*}
\tildepi_t&\leq \theta p_t\mathbf 1_{\{\theta_t <\psi(t)\}}+ ( \sum_{k=1}^\infty (u^{H^{(\tau,1)}}_{t,t+k} (c_{t+k})-u^{H^{(\tau,0)}}_{t,t+k} (c_{t+k}))p_t\mathbf 1_{\{\theta_t >\psi(t)\}}=\left (\psi(t)\wedge \theta\right) p_t.
\end{align*}

In conclusion, assuming that the agent is susceptible at time $t$, the agent:
\begin{itemize}
\item[-] buys insurance at time $t$ if
$$
\tildepi_t \leq p_t \left(\psi(t)\wedge\theta_t\right);
$$ 
\item[-] if diagnosed during the time interval $(t,t+1]$, the agent (i) buys the treatment if $\psi(t)\geq \theta_t$ and (ii) goes without treatment when $\psi(t)<\theta_t$. 
\end{itemize}
We can now reconnect this setting with the one in Section \ref{sect:Customers}. We assume that at each time $t$, there is a probability distribution $\mu_t:[0, 1]\times \mathbb R_+\to[0,1]$, characterizing the population of susceptible agents at time $t$, $S_t$. The first dimension represents the subjective probability of getting a diagnosis during $(t,t+1]$, conditionally on being susceptible at time $t$, while the second dimension represents the drug reservation price. Exactly as in the one-period setting, we find that, given the couple $(\theta_t,\tildepi_t)$, the population splits in three categories at time $t$, according to their decisions: a proportion of agents that buy insurance coverage:
$$
\muA_t=\mu_t\left(\{(p, \psi)\in[0,1]\times \mathbb R_+: px>y, \; p\psi >y \}\right)|_{x=\theta_t,y=\tildepi_t}, 
$$ another one  that instead buy the treatment out of pocket if diagnosed during $(t,t+1]$:
$$
\muT_t= \mu_t\left(\{(p, \psi)\in[0,1]\times \mathbb R_+: \psi>x, \; px \le y\}\right)|_{x=\theta_t,y=\tildepi_t}
$$ and the remaining proportion of agents,  untreated in case of a diagnosis:
$$
\muO_t=\mu_t\left(\{(p, \psi)\in[0,1]\times \mathbb R_+: \psi\le x, \; p\psi \le y\}\right)|_{x=\theta_t,y=\tildepi_t}.
$$  The notation is similar to the one in Section \ref{sect:Customers}, so that we skip further details here.
\subsection{Population dynamics and the producer-insurer game}\label{subs:dynpop}
Now that  the population response is modeled, it is possible to define the producer-insurer game in the dynamic setting. It should be impacted by the health decisions of the agents, but also possibly by exogenous factors such as aging, or improvement in screening. We only sketch the main ideas below.

At each time $t$ (beginning of each period), the two big players set a price for the drug and insurance, respectively. The information available to the players at time $t$ consists in:
\begin{itemize}
\item[-] the size of the population of susceptible $S_t$, and the distribution $\mu_t$ of agents characteristics $(p,\psi)$;
\item[-] the disease incidence rate in the period $(t,t+1]$, as estimated at time $t$, denoted $r_t$.
\end{itemize} 

One may assume---for instance---a Markov process of the SIR (susceptible-ill-removed) type for the population dynamics, and a Markovian evolution for $(r_t,\mu_t)$ as well. We do not extend the modeling (it goes beyond our objectives). Our point in this appendix is simply to underline that our one-period decision model for the population, despite its simplicity, is scalable to a multi-period setting, using as state processes, $(S_t,\mu_t,r_t)$. Our results for one period remain relevant in more complex settings.

\end{document}